\documentclass{llncs}

\usepackage{epsfig}
\usepackage{compress}
\usepackage{times}
\usepackage{gensymb}
\usepackage[english]{babel}
\usepackage{graphicx}
\usepackage{amssymb}
\usepackage{multirow}

\title{SEFE with No Mapping via \\Large Induced Outerplane Graphs in Plane Graphs}

\date{}

\author{Patrizio Angelini$^1$, William Evans$^2$, Fabrizio Frati$^3$, Joachim Gudmundsson$^3$
\institute{
$1$ Dipartimento di Ingegneria, Roma Tre University, Italy\\
\email{angelini@dia.uniroma3.it}\\
$2$  Department of Computer Science, University of British Columbia, Canada\\
\email{will@cs.ubc.ca}\\
$3$ School of Information Technologies, The University of Sydney, Australia\\
\email{brillo@it.usyd.edu.au}\\
\email{joachim.gudmundsson@sydney.edu.au}
}}

\newcommand{\sefenomap}{{\sc SefeNoMap~}}

\renewcommand{\int}{int}

\newcommand{\remove}[1]{}

\renewcommand{\int}{int}
\newcommand{\cl}{cl}

\renewenvironment{proof}
{{\bf Proof:}}{\hspace*{\fill}$\Box$\par\vspace{2mm}}

\newenvironment{proofx}
{{\bf Proof of Theorem~\ref{th:sefe-main}:}}{\hspace*{\fill}$\Box$\par\vspace{2mm}}

\begin{document}
\pagestyle{plain}

\maketitle

\begin{abstract}
We show that every $n$-vertex planar graph admits a simultaneous embedding with
no mapping and with fixed edges with any $(n/2)$-vertex planar graph. In order
to achieve this result, we prove that every $n$-vertex plane graph has an
induced outerplane subgraph containing at least $n/2$ vertices. Also, we show
that every $n$-vertex planar graph and every $n$-vertex planar partial
$3$-tree admit a
simultaneous embedding with no mapping and with fixed edges.
\end{abstract}

\section{Introduction} \label{se:introduction}

{\em Simultaneous embedding} is a flourishing area of research studying
topological and geometric properties of planar drawings of multiple graphs on
the same point set. The seminal paper in the area is the one of Bra{\ss} {\em et
al.}~\cite{bcd-spge-07}, in which two types of simultaneous embedding are
defined, namely {\em with mapping} and {\em with no mapping}. In the former
variant, a bijective mapping between the vertex sets of any two graphs $G_1$ and
$G_2$ to be drawn is part of the problem's input, and the goal is to construct a
planar drawing of $G_1$ and a planar drawing of $G_2$ so that corresponding
vertices are mapped to the same point. In the latter variant, the drawing
algorithm is free to map any vertex of $G_1$ to any vertex of $G_2$ (still the
$n$ vertices of $G_1$ and the $n$ vertices of $G_2$ have to be placed on the
same $n$ points). Simultaneous embeddings have been studied with respect to two
different drawing standards: In {\em geometric simultaneous embedding}, edges are required to be straight-line segments.
In {\em simultaneous embedding with fixed edges} (also known as {\sc Sefe}), edges can be arbitrary Jordan curves, but each edge that belongs to
two graphs $G_1$ and $G_2$ has to be represented by the same Jordan curve in the
drawing of $G_1$ and in the drawing of $G_2$.

Many papers deal with the problem of constructing geometric simultaneous
embeddings and simultaneous embeddings with fixed edges of pairs of planar
graphs in the variant {\em with mapping}. Typical considered problems include:
(i) determining notable classes of planar graphs that always or not always admit
a simultaneous embedding; (ii) designing algorithms for constructing
simultaneous embeddings within small area and with few bends on the edges; (iii)
determining the time complexity of testing the existence of a simultaneous
embedding for a given set of graphs. We refer the reader to the recent survey by
Blasi\"us, Kobourov, and Rutter~\cite{bkr-sepg-13}.

In contrast to the large number of papers dealing with simultaneous embedding
{\em with mapping}, little progress has been made on the {\em no mapping}
version of the problem.
Bra{\ss} {\em et
al.}~\cite{bcd-spge-07} showed that any planar graph admits a geometric
simultaneous embedding with no mapping with any number of outerplanar graphs.
They left open the following attractive question: Do every two $n$-vertex planar
graphs admit a geometric simultaneous embedding with no mapping?

In this paper we initiate the study of simultaneous
embeddings with fixed edges and no mapping, called
\sefenomap for brevity. In this setting, the natural counterpart of the Bra{\ss}
{\em et al.}~\cite{bcd-spge-07} question reads as follows: Do every two
$n$-vertex planar graphs admit a \sefenomap?

Since answering this question seems to be an elusive goal, we tackle
the following generalization of the problem: What is the largest $k\leq n$ such
that every $n$-vertex planar graph and every $k$-vertex planar graph admit a
\sefenomap? That is: What is the largest $k\leq n$ such that every
$n$-vertex planar graph $G_1$ and every $k$-vertex planar graph $G_2$ admit two
planar drawings $\Gamma_1$ and $\Gamma_2$ with their vertex sets mapped to point
sets $P_1$ and $P_2$, respectively, so that $P_2 \subseteq P_1$ and so that if
edges $e_1$ of $G_1$ and $e_2$ of $G_2$ have their end-vertices mapped to the
same two points $p_a$ and $p_b$, then $e_1$ and $e_2$ are represented by the
same Jordan curve in $\Gamma_1$ and in $\Gamma_2$? We prove that $k\geq n/2$:

\begin{theorem} \label{th:sefe-main}
Every $n$-vertex planar graph and every $(n/2)$-vertex planar graph have a
\sefenomap.
\end{theorem}

Observe that the previous theorem would be easily proved if $n/2$ were replaced
with $n/4$: First, consider an $(n/4)$-vertex independent set $I$ of any
$n$-vertex planar graph $G_1$ (which always exists, as a consequence of the four
color theorem~\cite{ah-epfci-77,ahk-epfcii-77}). Then, construct any planar
drawing $\Gamma_1$ of $G_1$, and let $P(I)$ be the point set on which the
vertices of $I$ are mapped in $\Gamma_1$. Finally, construct a planar drawing
$\Gamma_2$ of any $(n/4)$-vertex planar graph $G_2$ on point set $P(I)$ (e.g.
using Kaufmann and Wiese's technique~\cite{kw-evpfbspg-02}). Since $I$ is an
independent set, any bijective mapping between the vertex set of $G_2$ and $I$
ensures that $G_1$ and $G_2$ share no edges. Thus, $\Gamma_1$ and $\Gamma_2$ are
a \sefenomap of $G_1$ and $G_2$.

In order to get the $n/2$ bound, we study the problem of finding a large
induced outerplane graph in a plane graph.
A {\em plane graph} is a
planar graph together with a {\em plane embedding}, that is, an equivalence class
of planar drawings, where two planar drawings $\Gamma_1$ and $\Gamma_2$ are
equivalent if: (1) each vertex has the same {\em rotation scheme} in $\Gamma_1$
and in $\Gamma_2$, i.e., the same clockwise order of the edges incident to it;
(2) each face has the same {\em facial cycles} in $\Gamma_1$ and in $\Gamma_2$,
i.e., it is delimited by the same set of cycles; and (3) $\Gamma_1$ and
$\Gamma_2$ have the same {\em outer face}. An {\em outerplane graph} is a graph
together with an {\em outerplane embedding}, that is a plane embedding where all
the vertices are incident to the outer face. An {\em outerplanar graph} is a
graph that admits an outerplane embedding; a plane embedding of an outerplanar
graph is not necessarily outerplane. Consider a plane graph $G$ and a subset
$V'$ of its vertex set. The {\em induced plane graph} $G[V']$ is the subgraph of
$G$ induced by $V'$ together with the plane embedding {\em inherited} from $G$,
i.e., the embedding obtained from the plane embedding of $G$ by removing all the
vertices and edges not in $G[V']$.
We show the following result:

\begin{theorem} \label{th:outerplane-main}
Every $n$-vertex plane graph $G(V,E)$ has a vertex set $V'\subseteq V$ with
$|V'|\geq n/2$ such that $G[V']$ is an outerplane graph.
\end{theorem}

Theorem~\ref{th:outerplane-main} and the results of Gritzmann {\em et al.}~\cite{gmpp-eptvs-91} yield a proof of Theorem~\ref{th:sefe-main}, as follows:

\begin{figure}[tb]
\begin{center}
\begin{tabular}{c c c}
\mbox{\includegraphics[scale=0.47]{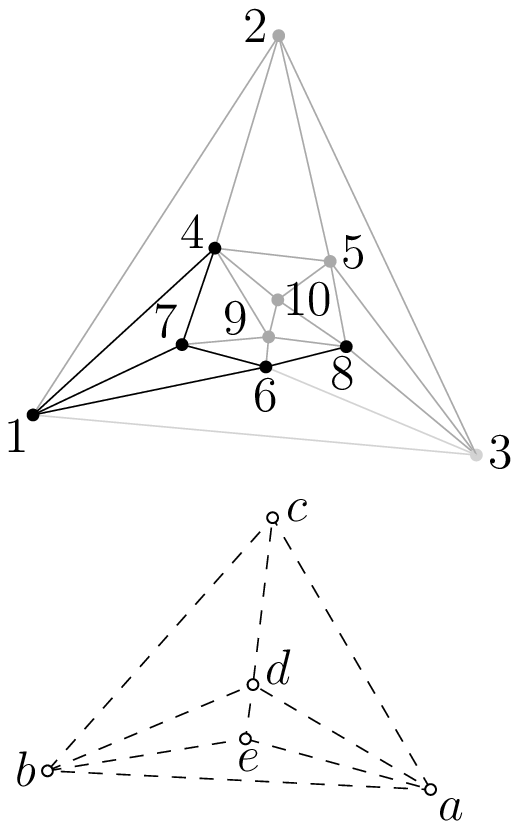}} \hspace{2mm} &
\mbox{\includegraphics[scale=0.47]{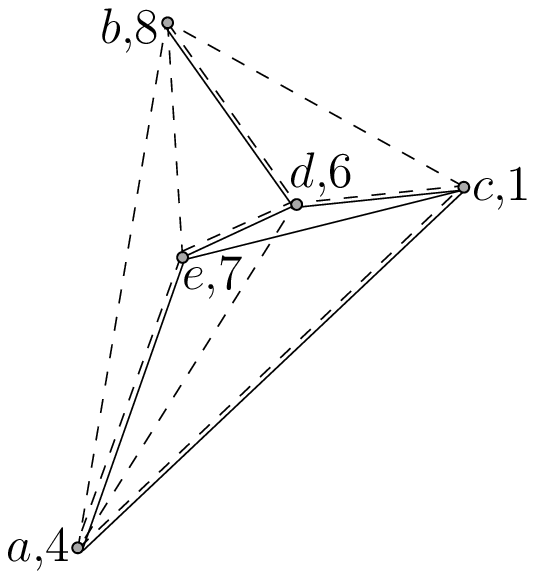}} \hspace{2mm} &
\mbox{\includegraphics[scale=0.47]{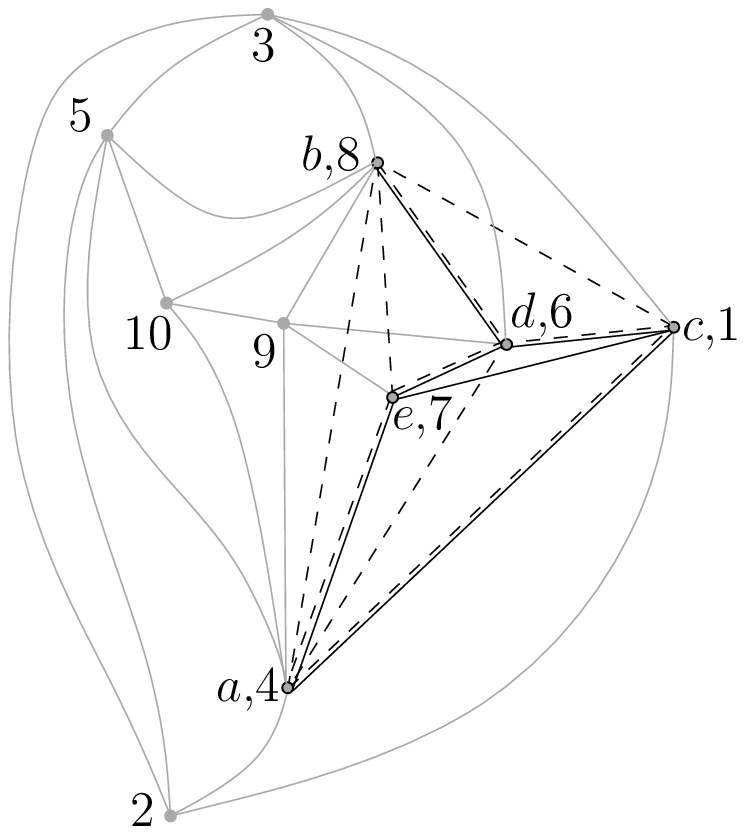}}\\
(a) \hspace{2mm} & (b) \hspace{2mm} & (c)
\end{tabular}
\caption{(a) A $10$-vertex planar graph $G_1$ (solid lines) and a $5$-vertex planar graph $G_2$ (dashed lines). A $5$-vertex induced outerplane graph $G_1[V']$ in $G_1$ is colored black. Vertices and edges of $G_1$ not in $G_1[V']$ are colored gray. (b) A straight-line planar drawing $\Gamma(G_2)$ of $G_2$ with no three collinear vertices, together with a straight-line planar drawing of $G_1[V']$ on the point set $P_2$ defined by the vertices of $G_2$ in $\Gamma(G_2)$. (c) A \sefenomap of $G_1$ and $G_2$.}
\label{fig:illustration}
\end{center}
\end{figure}

\begin{proofx}
Consider any $n$-vertex plane graph $G_1$ and any $(n/2)$-vertex plane graph $G_2$ (see Fig.~\ref{fig:illustration}(a)). Let $\Gamma(G_2)$ be any straight-line planar drawing of $G_2$ in which no three vertices are collinear. Denote by $P_2$ the set of $n/2$ points to which the vertices of $G_2$ are mapped in $\Gamma(G_2)$. Consider any vertex subset $V'\subseteq V(G_1)$ such that $G_1[V']$ is an outerplane graph. Such a set exists by Theorem~\ref{th:outerplane-main}. Construct a straight-line planar drawing $\Gamma(G_1[V'])$ of $G_1[V']$ in which its vertices are mapped to $P_2$ so that the resulting drawing has the same (outerplane) embedding as $G_1[V']$. Such a drawing exists by results of Gritzmann {\em et al.}~\cite{gmpp-eptvs-91}; also it can found efficiently by results of Bose~\cite{b-eogps-02} (see Fig.~\ref{fig:illustration}(b)). Construct any planar drawing $\Gamma(G_1)$ of $G_1$ in which the drawing of $G_1[V']$ is $\Gamma(G_1[V'])$. Such a drawing exists, given that $\Gamma(G_1[V'])$ is a planar drawing of a plane subgraph $G_1[V']$ of $G_1$ preserving the embedding of $G_1[V']$ in $G_1$ (see Fig.~\ref{fig:illustration}(c)). Both $\Gamma(G_1)$ and $\Gamma(G_2)$ are planar, by construction. Also, the only edges that are possibly shared by $G_1$ and $G_2$ are those between two vertices that are mapped to $P_2$. However, such edges are drawn as straight-line segments both in $\Gamma(G_1)$ and in $\Gamma(G_2)$. Thus, $\Gamma(G_1)$ and $\Gamma(G_2)$ are a \sefenomap of $G_1$ and $G_2$.
\end{proofx}

By the standard observation that the vertices in the odd (or even) levels of a breadth-first search tree of a planar graph induce an \emph{outerplanar} graph,
we know that $G$ has an induced outerplanar graph with at least $n/2$ vertices. However, since its embedding in $G$ may not be outerplane, this seems insufficient to prove the existence of a \sefenomap of every $n$-vertex and every $(n/2)$-vertex planar graph.

Theorem~\ref{th:outerplane-main} might be of independent interest, as it is
related to (in fact it is a weaker version of) one of the most famous and
long-standing graph theory conjectures:

\begin{conjecture} \label{conj:induced-forests} {\em (Albertson and Berman
1979~\cite{ab-cpg-79})}
Every $n$-vertex planar graph $G(V,E)$ has a vertex set $V'\subseteq V$ with
$|V'|\geq n/2$ such that $G[V']$ is a forest.
\end{conjecture}

Conjecture~\ref{conj:induced-forests} would prove the existence of an
$(n/4)$-vertex independent set in a planar graph without using the four color
theorem~\cite{ah-epfci-77,ahk-epfcii-77}. The best known partial result related
to Conjecture~\ref{conj:induced-forests} is that every planar graph has a vertex
subset with $2/5$ of its vertices inducing a forest, which is a consequence of
the {\em acyclic $5$-colorability} of planar graphs~\cite{Borodin79}. Variants
of the conjecture have also been studied such that the planar graph in which the
induced forest has to be found is bipartite~\cite{aw-mifpg-87}, or is
outerplanar~\cite{h-iftog-90}, or such that each connected component of the
induced forest is required to be a path~\cite{p-milfog-04,p-lvapg-90}.


The topological structure of an outerplane graph is arguably much closer to that of a forest than the one of a non-outerplane graph. Thus the importance of Conjecture~\ref{conj:induced-forests} may justify the study of induced outerplane graphs in plane graphs in its own right.

To complement the results of the paper, we also show the following:

\begin{theorem} \label{th:plane3trees-main}
Every $n$-vertex planar graph and every $n$-vertex planar parital $3$-tree have a \sefenomap.
\end{theorem}

\section{Proof of Theorem~\ref{th:outerplane-main}} \label{se:outerplane}

In this section we prove Theorem~\ref{th:outerplane-main}. We assume that the input graph $G$ is a {\em maximal} plane graph, that is, a plane graph such that no edge can be added to it while maintaining planarity. In fact, if $G$ is not maximal, then dummy edges can be added to it in order to make it a maximal plane graph $G'$. Then, the vertex set $V'$ of an induced outerplane graph $G'[V']$ in $G'$ induces an outerplane graph in $G$, as well.

\begin{figure}[tb]
\begin{center}
\begin{tabular}{c  c c}
\mbox{\includegraphics[scale=0.25]{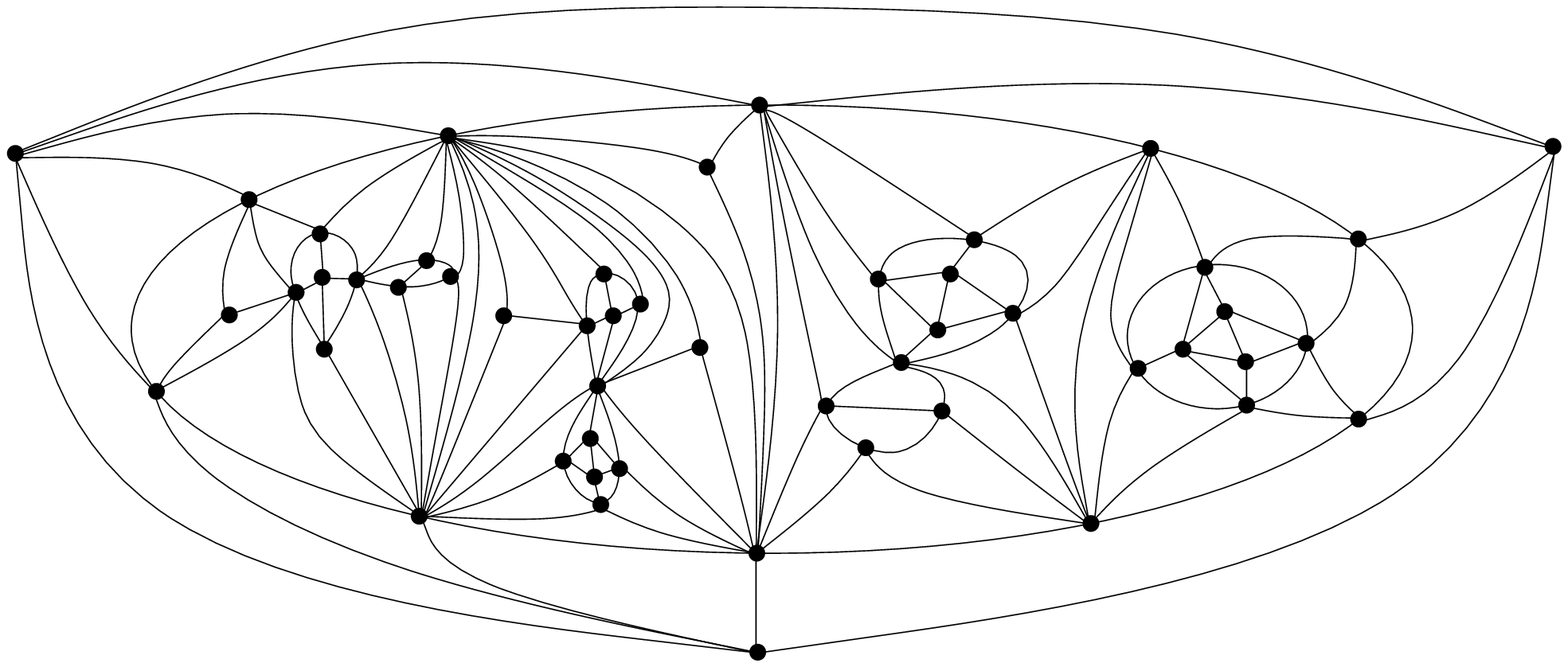}} \hspace{7mm} &
\mbox{\includegraphics[scale=0.32]{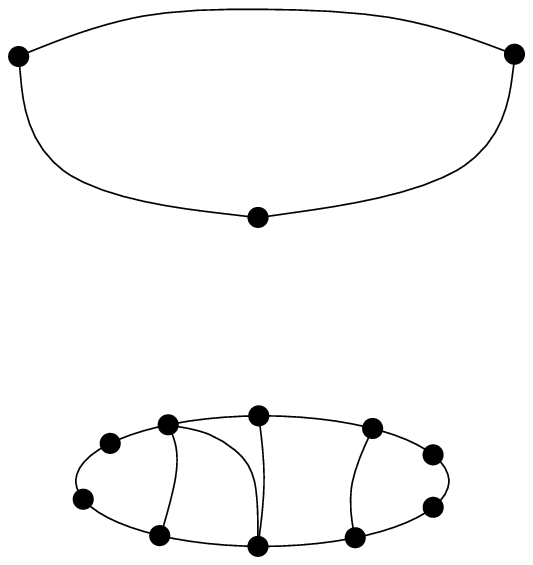}} \hspace{7mm} &
\mbox{\includegraphics[scale=0.32]{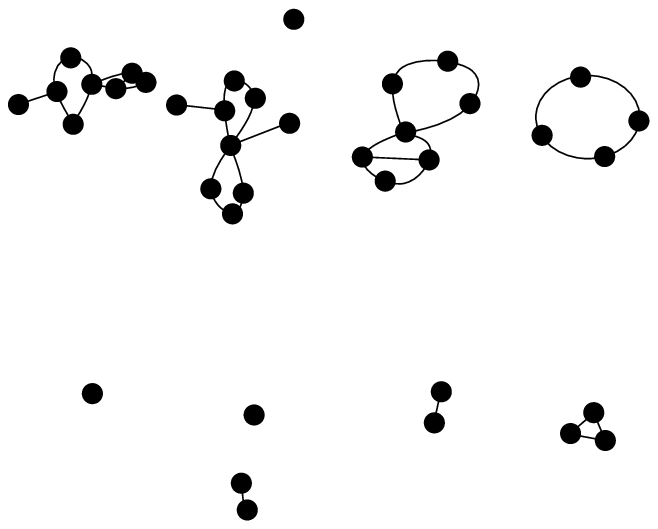}}\\
(a) \hspace{7mm}& (b) \hspace{7mm} & (c)
\end{tabular}
\caption{(a) A maximal plane graph $G$ with outerplanarity $4$. (b) Graphs $G[V_1]$ (on the top) and $G[V_2]$ (on the bottom). (c) Graphs $G[V_3]$ (on the top) and $G[V_4]$ (on the bottom).}
\label{fig:outerplanelevels}
\end{center}
\end{figure}

Let $G^*_1=G$ and, for any $i\geq 1$, let $G^*_{i+1}$ be the plane graph
obtained by removing from $G^*_i$ the set $V_i$ of vertices incident to the
outer face of $G^*_i$ and their incident edges. Vertex set $V_i$ is the {\em
$i$-th outerplane level} of $G$. Denote by $k$ the maximum index such that $V_k$
is non-empty; then $k$ is the {\em outerplanarity} of $G$. For any $1\leq i\leq
k$, graph $G[V_i]$ is a (not necessarily connected) outerplane graph and graph
$G^*_i$ is a (not necessarily connected) {\em internally-triangulated} plane
graph, that is, a plane graph whose internal faces are all triangles. See
Fig.~\ref{fig:outerplanelevels}. For $1\leq i \leq k$, denote by
$H^*_{i,1},\dots,H^*_{i,h_i}$ the connected components of $G^*_i$ and, for
$1\leq j \leq h_i$, denote by $H_{i,j}$ the outerplane graph induced by the
vertices incident to the outer face of $H^*_{i,j}$. Since $G$ is maximal, for
any $1\leq i \leq k$ and for any internal face $f$ of $G[V_i]$, at most one
connected component of $G^*_{i+1}$ lies inside $f$.

A {\em $2$-coloring} $\psi=(W^*,B^*)$ of a graph $H^*$ is a partition of the vertex
set $V(H^*)$ into two sets $W^*$ and $B^*$. We say that the vertices in $W^*$
are \emph{white} and the ones in $B^*$ are \emph{black}. Given a $2$-coloring
$\psi=(W^*,B^*)$ of a plane graph $H^*$, the subgraph $H^*[W^*]$ of $H^*$ is
{\em strongly outerplane} if it is outerplane and it contains no black vertex
inside any of its internal faces. We define the \emph{surplus} of $\psi$ as
$s(H^*,\psi)=|W^*|-|B^*|$.

A {\em cutvertex} in a connected graph $H^*$ is a vertex whose removal
disconnects $H^*$. A \emph{maximal $2$-connected component} of $H^*$, also
called a \emph{block} of $H^*$, is an induced subgraph $H^*[V']$ of $H^*$ such
that $H^*[V']$ is $2$-connected and there exists no $V''\subseteq
V(H^*)$ where $V'\subset V''$ and $H^*[V'']$ is $2$-connected. The {\em
block-cutvertex tree} ${\cal BC}(H^*)$ of $H^*$ is a tree that represents the
arrangement of the blocks of $H^*$ (see Figs.~\ref{fig:bctree}(a)
and~\ref{fig:bctree}(b)). Namely, ${\cal BC}(H^*)$ contains a \emph{${\cal
B}$-node} for each block of $H^*$ and a \emph{${\cal C}$-node} for each
cutvertex of $H^*$; further, there is an edge between a ${\cal B}$-node $b$ and
a ${\cal C}$-node $c$ if $c$ is a vertex of $b$. Given a $2$-coloring
$\psi=(W^*,B^*)$ of $H^*$, the {\em contracted block-cutvertex tree} ${\cal
CBC}(H^*,\psi)$ of $H^*$ is the tree obtained from ${\cal BC}(H^*)$ by
identifying all the ${\cal B}$-nodes that are adjacent to the same black
cut-vertex $c$, and by removing $c$ and its incident edges (see
Fig.~\ref{fig:bctree}(c)). Each node of ${\cal CBC}(H^*,\psi)$ is either a
${\cal C}$-node $c$ or a ${\cal BU}$-node $b$. In the former case, $c$
corresponds to a white ${\cal C}$-node in ${\cal BC}(H^*)$. In the latter case,
$b$ corresponds to a maximal connected subtree ${\cal BC}(H^*(b))$ of ${\cal
BC}(H^*)$ only containing ${\cal B}$-nodes and black ${\cal C}$-nodes. The {\em
subgraph $H^*(b)$ of $H^*$ associated with a ${\cal BU}$-node $b$} is the union
of the blocks of $H^*$ corresponding to ${\cal B}$-nodes in ${\cal BC}(H^*(b))$.
Finally, we denote by $H(b)$ the outerplane graph induced by the vertices
incident to the outer face of $H^*(b)$. We have the following:

\begin{figure}[tb]
\begin{center}
\begin{tabular}{c c c}
\mbox{\includegraphics[scale=0.5]{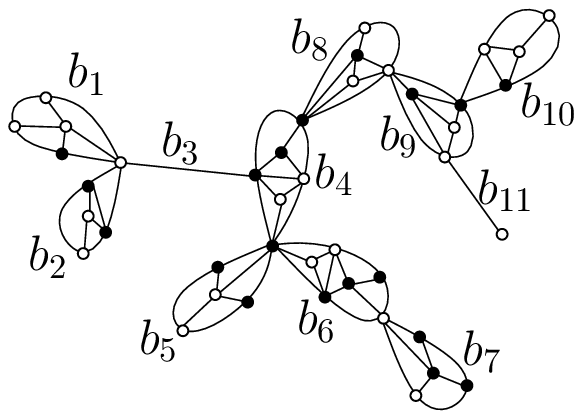}} \hspace{3mm} &
\mbox{\includegraphics[scale=0.55]{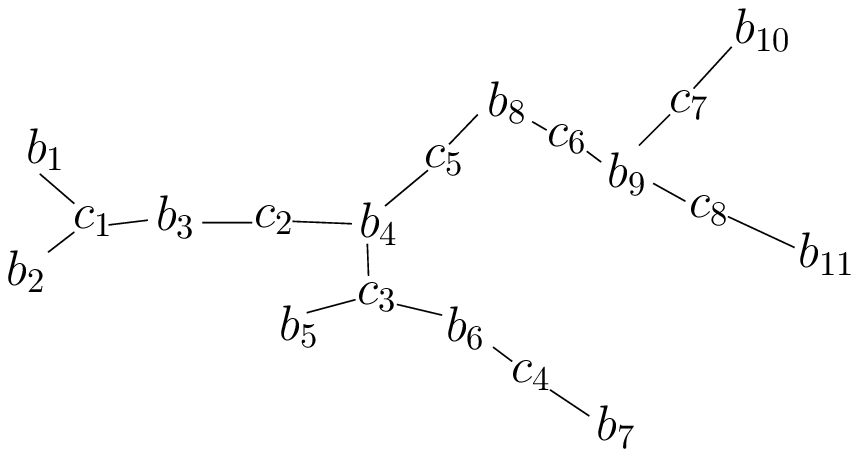}} \hspace{3mm} &
\mbox{\includegraphics[scale=0.55]{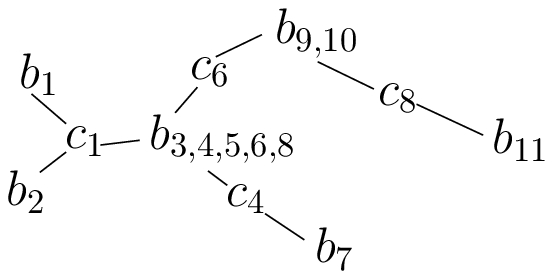}}\\
(a) \hspace{3mm}& (b) \hspace{3mm} & (c)
\end{tabular}
\caption{(a) A connected internally-triangulated plane graph $H^*$ with a $2$-coloring $\psi$, (b) the block-cutvertex tree ${\cal BC}(H^*)$, and (c) the contracted block-cutvertex tree ${\cal CBC}(H^*,\psi)$.}
\label{fig:bctree}
\end{center}
\end{figure}

\begin{lemma} \label{th:inductive}
For any $1\leq i\leq k$ and any $1\leq j\leq h_i$, there exists a $2$-coloring $\psi=(W^*_{i,j},B^*_{i,j})$ of $H^*_{i,j}$ such that:

\begin{itemize}
\item[(1)] the subgraph $H^*_{i,j}[W^*_{i,j}]$ of $H^*_{i,j}$ induced by $W^*_{i,j}$ is strongly outerplane; and
\item[(2)] for any ${\cal BU}$-node $b$ in ${\cal CBC}(H^*_{i,j},\psi)$, one of the following holds:
\begin{itemize}
\item[(a)] $s(H^*_{i,j}(b),\psi)\geq |W^*_{i,j} \cap V(H_{i,j}(b))|+1$;
\item[(b)] $s(H^*_{i,j}(b),\psi)= |W^*_{i,j} \cap V(H_{i,j}(b))|$ and there exists an edge with white end-vertices incident to the outer face of $H^*_{i,j}(b)$; or
\item[(c)] $s(H^*_{i,j}(b),\psi)=1$ and $H^*_{i,j}(b)$ is a single vertex.
\end{itemize}
\end{itemize}
\end{lemma}

Lemma~\ref{th:inductive} implies Theorem~\ref{th:outerplane-main} as follows: Since $G$ is a maximal plane graph, $G^*_1$ has one $2$-connected component, hence $H^*_{1,1}(b)=H^*_{1,1}=G^*_1=G$. By Lemma~\ref{th:inductive}, there exists a $2$-coloring $\psi=(W,B)$ of $G$ such that $G[W]$ is an outerplane graph and $|W|-|B|\geq |W \cap V_1| \geq 0$, hence  $|W| \geq n/2$.

We emphasize that Lemma~\ref{th:inductive} shows the existence of a large induced subgraph $H^*_{i,j}[W^*_{i,j}]$ of $H^*_{i,j}$ satisfying an even stronger property than just being outerplane; namely, the $2$-coloring $\psi=(W^*_{i,j},B^*_{i,j})$ is such that $H^*_{i,j}[W^*_{i,j}]$ is outerplane and contains no vertex belonging to $B^*_{i,j}$ in any of its internal faces.

In order to prove Lemma~\ref{th:inductive}, we start by showing some sufficient conditions for a $2$-coloring to induce a strongly outerplane graph in $H^*_{i,j}$. We first state a lemma arguing that a $2$-coloring $\psi$ of $H^*_{i,j}$ satisfies Condition (1) of Lemma~\ref{th:inductive} if and only if it satisfies the same condition ``inside each internal face'' of $H_{i,j}$. For any face $f$ of $H_{i,j}$, we denote by $C_f$ the cycle delimiting $f$; also, we denote by $H^*_{i,j}[W^*_{i,j}(f)]$ the subgraph of $H^*_{i,j}$ induced by the white vertices inside or belonging to $C_f$.

\begin{lemma} \label{le:single-faces}
Let $\psi=(W^*_{i,j},B^*_{i,j})$ be a $2$-coloring of $H^*_{i,j}$. Assume that, for each internal face $f$ of $H_{i,j}$, graph $H^*_{i,j}[W^*_{i,j}(f)]$ is strongly outerplane. Then, $H^*_{i,j}[W^*_{i,j}]$ is strongly outerplane.
\end{lemma}

\begin{proof}
Suppose, for a contradiction, that $H^*_{i,j}[W^*_{i,j}]$ is not strongly outerplane. Then, it contains a simple cycle $C$ that contains in its interior some vertex $x$ in $H^*_{i,j}$. Assume, w.l.o.g., that $C$ is minimal, that is, there exists no cycle $C'$ that contains $x$ in its interior such that $|V(C')|\subset |V(C)|$. By hypothesis, $C$ is not a subgraph of $H^*_{i,j}[W^*_{i,j}(f)]$, for any internal face $f$ of $H_{i,j}$. Consider a maximal path $P$ in $C$ all of whose edges belong to $H^*_{i,j}[W^*_{i,j}(f)]$, for some internal face $f$ of $H_{i,j}$. Let $u$ and $v$ be the end-vertices of $P$; also, let $w$ be the vertex adjacent to $v$ in $C$ and not belonging to $P$. See Fig.~\ref{fig:single-face}. Vertex $v$ belongs to $C_f$, as otherwise edge $(v,w)$ would cross $C_f$, given that $w$ is not inside nor belongs to $C_f$, by the maximality of $P$. Let $v'$ and $v''$ be the vertices adjacent to $v$ on $C_f$. We have that the path $C - P$ obtained from $C$ by removing the edges and the internal vertices of $P$ contains $v'$ or $v''$. In fact, if that's not the case, $C - P$ would pass twice through $v$, which contradicts the fact that $C$ is a simple cycle. However, if $C$ contains one of $v'$ or $v''$, say $v'$, then $v'$ is a white vertex, thus edge $(v,v')$ splits $C$ into two cycles $C'$ and $C''$, with $|V(C')|\subset |V(C)|$ and $|V(C'')|\subset |V(C)|$, one of which contains $x$ in its interior, thus contradicting the minimality of $C$.
\begin{figure}[tb]
\begin{center}
\mbox{\includegraphics[scale=0.55]{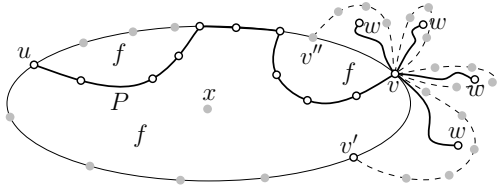}}
\caption{Illustration for the proof of Lemma~\ref{le:single-faces}. The thick solid line represents $P$ together with edge $(v,w)$. The thin solid lines represent the edges of $C_f$ not in $P$. The dashed lines represent some edges of $H_{i,j}$ not in $C_f$. The color of the gray vertices is not important for the proof.}
\label{fig:single-face}
\end{center}
\end{figure}
\end{proof}


An internal face $f$ of $H_{i,j}$ is {\em empty} if it contains no vertex of $G^*_{i+1}$ in its interior. Also, for a $2$-coloring $\psi$ of $H^*_{i,j}$, an internal face $f$ of $H_{i,j}$ is {\em trivial} if it contains in its interior a connected component $H^*_{i+1,k}$ of $G^*_{i+1}$ that is a single white vertex or such that all the vertices incident to the outer face of $H^*_{i+1,k}$ are black. We have the following.

\begin{lemma} \label{le:trivial-faces}
Let $\psi=(W^*_{i,j},B^*_{i,j})$ be a $2$-coloring of $H^*_{i,j}$ and let $f$ be a trivial face of $H_{i,j}$. Let $H^*_{i+1,k}$ be the connected component of $G^*_{i+1}$ in $f$'s interior. If $H^*_{i+1,k}[W^*_{i,j}]$ is strongly outerplane and if $C_f$ contains at least one black vertex, then $H^*_{i,j}[W^*_{i,j}(f)]$ is strongly outerplane.
\end{lemma}

\begin{proof}
Suppose, for a contradiction, that $H^*_{i,j}[W^*_{i,j}(f)]$ is not strongly outerplane. Then, it contains a simple cycle $C$ that contains in its interior some vertex $x$ in $H^*_{i,j}$. Cycle $C$ contains at least one vertex of $C_f$ by the assumption that $H^*_{i+1,k}[W^*_{i,j}]$ is outerplane. Also, $C$ does not coincide with $C_f$, since $C_f$ contains at least one black vertex. Then, $C$ contains vertices of $C_f$ and vertices internal to $C_f$. This provides a contradiction in the case in which $H^*_{i+1,k}$ is a single white vertex, as no other vertex is internal to any cycle in $H^*_{i,j}[W^*_{i,j}(f)]$, and it provides a contradiction in the case in which all the vertices incident to the outer face of $H^*_{i+1,k}$ are black, as no edge connects a white vertex of $C_f$ with a white vertex internal to $C_f$.
\end{proof}


We now prove Lemma~\ref{th:inductive} by induction on the outerplanarity of $H^*_{i,j}$.

In the base case, the outerplanarity of $H^*_{i,j}$ is $1$; then, color white all the vertices of $H^*_{i,j}$. Since the outerplanarity of  $H^*_{i,j}$ is $1$, then $H^*_{i,j}[W^*_{i,j}]=H^*_{i,j}$ is an outerplane graph, thus satisfying Condition (1) of Lemma~\ref{th:inductive}. Also, consider any ${\cal BU}$-node $b$ in the contracted block-cutvertex tree ${\cal CBC}(H^*_{i,j},\psi)$ (which coincides with the block-cutvertex tree ${\cal BC}(H^*_{i,j})$, given that all the vertices of $H^*_{i,j}$ are white). All the vertices of $H^*_{i,j}(b)$ are white, hence either Condition (2b) or Condition (2c) of Lemma~\ref{th:inductive} is satisfied, depending on whether $H^*_{i,j}(b)$ has or does not have an edge, respectively.

In the inductive case, the outerplanarity of $H^*_{i,j}$ is greater than $1$.

First, we inductively construct a $2$-coloring
$\psi_k=(W^*_{i+1,k},B^*_{i+1,k})$, satisfying the conditions of
Lemma~\ref{th:inductive}, of each connected component $H^*_{i+1,k}$ of
$G^*_{i+1}$, for $1\leq k \leq h_{i+1}$. The $2$-coloring $\psi$ of $H^*_{i,j}$
is such that each connected component $H^*_{i+1,k}$ of $G^*_{i+1}$ that lies
inside an internal face of $H_{i,j}$ ``maintains'' the coloring $\psi_k$, i.e.,
a vertex of $H^*_{i+1,k}$ is white in $\psi$ if and only if it is white in
$\psi_k$. Then, in order to determine $\psi$, it suffices to describe how to
color the vertices of $H_{i,j}$.

Second, we look at the internal faces of $H_{i,j}$ one at a time. When we look at a face $f$, we determine a set $B_f$ of vertices of $C_f$ that are colored black. This is done in such a way that the graph $H^*_{i,j}[W^*_{i,j}(f)]$ is strongly outerplane even if we color white all the vertices in $V(C_f)\setminus B_f$. By Lemma~\ref{le:single-faces}, a $2$-coloring of $H^*_{i,j}$ such that $H^*_{i,j}[W^*_{i,j}(f)]$ is strongly outerplane for every internal face $f$ of $H_{i,j}$ is such that $H^*_{i,j}[W^*_{i,j}]$ is strongly outerplane. We remark that, when a set $B_f$ of vertices of $C_f$ are colored black, the vertices in $V(C_f)\setminus B_f$ are not necessarily colored white, as a vertex in $V(C_f)\setminus B_f$ might belong to the set $B_{f'}$ of vertices that are colored black for a face $f'\neq f$ of $H_{i,j}$. In fact, only after the set $B_f$ of vertices of $C_f$ are colored black for {\em every} internal face $f$ of $H_{i,j}$, are the remaining uncolored vertices in $H_{i,j}$ colored white.

We now describe in more detail how to color the vertices of $H_{i,j}$. We show an algorithm, that we call {\em algorithm cycle-breaker}, that associates a set $B_f$ to each internal face $f$ of $H_{i,j}$ as follows.

{\bf Empty faces:} For any empty face $f$ of $H_{i,j}$, let $B_f=\emptyset$.

{\bf Trivial faces:} While there exists a vertex $v^*_{1,2}$ incident to two
trivial faces $f_1$ and $f_2$ of $H_{i,j}$ to which no sets $B_{f_1}$ and
$B_{f_2}$ have been associated yet, respectively, let
$B_{f_1}=B_{f_2}=\{v^*_{1,2}\}$. When no such vertex exists, for any trivial
face $f$ of $H_{i,j}$ to which no set $B_f$ has been associated yet, let $v$ be
any vertex of $C_f$ and let $B_f=\{v\}$.

{\bf Non-trivial non-empty faces:} Consider any non-trivial non-empty internal
face $f$ of $H_{i,j}$. Denote by $H^*_{i+1,k}$ the connected component of
$G^*_{i+1}$ inside $f$. By induction, for any ${\cal BU}$-node $b$ in the
contracted block-cutvertex tree ${\cal CBC}(H^*_{i+1,k},\psi_k)$, it holds
$s(H^*_{i+1,k}(b),\psi_k)\geq |W^*_{i+1,k} \cap V(H_{i+1,k}(b))|+1$, or
$s(H^*_{i+1,k}(b),\psi_k)= |W^*_{i+1,k} \cap V(H_{i+1,k}(b))|$ and there exists
an edge incident to the outer face of $H^*_{i+1,k}(b)$ whose both end-vertices
are white.

We repeatedly perform the following actions: (i) We pick any ${\cal BU}$-node
$b$ that is a leaf in ${\cal CBC}(H^*_{i+1,k},\psi_k)$; (ii) we insert some
vertices of $C_f$ in $B_f$, based on the structure and the coloring of
$H^*_{i+1,k}(b)$; and (iii) we remove $b$ from ${\cal CBC}(H^*_{i+1,k},\psi_k)$,
possibly also removing its adjacent cutvertex, if it has degree one. We describe
in more detail action (ii).

For every white vertex $u$ incident to the outer face of $H^*_{i+1,k}(b)$, we
define the {\em rightmost neighbor $r(u,b)$ of $u$ in $C_f$ from $b$} as
follows. Denote by $u'$ the vertex following $u$ in the clockwise order of the
vertices along the cycle delimiting the outer face of $H^*_{i+1,k}(b)$. Vertex
$r(u,b)$ is the vertex preceding $u'$ in the clockwise order of the neighbors of
$u$. Observe that, since $H^*_{i,j}$ is internally-triangulated, then $r(u,b)$
belongs to $C_f$. Also, $r(u,b)$ is well-defined because $u$ is not a cutvertex
(in fact, it might be a cutvertex of $H^*_{i+1,k}$, but it is not a cutvertex of
$H^*_{i+1,k}(b)$, since such a graph contains no white cut-vertex).

Suppose that $s(H^*_{i+1,k}(b),\psi_k)\geq |W^*_{i+1,k} \cap V(H_{i+1,k}(b))|+1$. Then, for every white vertex $u$ incident to the outer face of $H^*_{i+1,k}(b)$, we add $r(u,b)$ to $B_f$.

Suppose that $s(H^*_{i+1,k}(b),\psi_k)= |W^*_{i+1,k} \cap V(H_{i+1,k}(b))|$ and there exists an edge $(v,v')$ incident to the outer face of $H^*_{i+1,k}(b)$ such that $v$ and $v'$ are white. Assume, w.l.o.g., that $v'$ follows $v$ in the clockwise order of the vertices along the cycle delimiting the outer face of $H^*_{i+1,k}(b)$. Then, for every white vertex $u\neq v$ incident to the outer face of $H^*_{i+1,k}(b)$, we add $r(u,b)$ to $B_f$.

After the execution of algorithm cycle-breaker, a set $B_f$ has been defined for every internal face $f$ of $H_{i,j}$. Then, color black all the vertices in $\bigcup_{f} B_f$, where the union is over all the internal faces $f$ of $H_{i,j}$. Also, color white all the vertices of $H_{i,j}$ that are not colored black. Denote by $\psi=(W^*_{i,j},B^*_{i,j})$ the resulting coloring of $H^*_{i,j}$. We have the following lemma, that completes the induction, and hence the proof of Lemma~\ref{th:inductive}.

\begin{lemma} \label{le:correctness}
Coloring $\psi$ satisfies Conditions (1) and (2) of Lemma~\ref{th:inductive}.
\end{lemma}

\begin{proof}
{\em We prove that $\psi$ satisfies Condition (1) of Lemma~\ref{th:inductive}}. Namely, we prove that, for every internal face $f$ of $H_{i,j}$, graph $H^*_{i,j}[W^*_{i,j}(f)]$ is strongly outerplane. By Lemma~\ref{le:single-faces}, this implies that graph $H^*_{i,j}[W^*_{i,j}]$ is strongly outerplane.

For any {\bf empty face} $f$ of $H_{i,j}$, graph $H^*_{i,j}[W^*_{i,j}(f)]$ is strongly outerplane, as no vertex of $H^*_{i,j}$ is in the interior of $f$.

By construction, for any {\bf trivial face} $f$ of $H_{i,j}$, there is a black vertex in $C_f$; hence, by Lemma~\ref{le:trivial-faces}, graph $H^*_{i,j}[W^*_{i,j}(f)]$ is strongly outerplane.

Let $f$ be any {\bf non-empty non-trivial internal face} of $H_{i,j}$. Denote by $H^*_{i+1,k}$ the connected component of $G^*_{i+1}$ inside $f$. Such a component exists because $f$ is non-empty. Suppose, for a contradiction, that $H^*_{i,j}[W^*_{i,j}(f)]$ is not strongly outerplane. Then, it contains a simple cycle $C$ that contains in its interior some vertex $x$ in $H^*_{i,j}$. Assume, w.l.o.g., that $C$ is minimal, that is, there exists no cycle $C'$ in $H^*_{i,j}[W^*_{i,j}(f)]$ that contains $x$ in its interior and such that $|V(C')|\subset |V(C)|$.

Cycle $C$ contains at least one vertex of $C_f$, since by induction $H^*_{i+1,k}[W^*_{i,j}]$ is strongly outerplane. Suppose, for a contradiction, that $C$ coincides with $C_f$. Since $f$ is non-trivial, when algorithm cycle-breaker picks the first leaf $b$ in ${\cal CBC}(H^*_{i+1,k},\psi_k)$, there exists at least one vertex $u$ incident to the outer face of $H^*_{i+1,k}(b)$ that is white. Then, either $r(u,b)\in V(C_f)$ is black, thus obtaining a contradiction, or there exists an edge $(u,u')$ incident to the outer face of $H^*_{i+1,k}(b)$ such that $u$ and $u'$ are both white, where $u'$ follows $u$ in the clockwise order of the vertices along the cycle delimiting the outer face of $H^*_{i+1,k}(b)$. Then $r(u',b)\in V(C_f)$ is black, thus obtaining a contradiction. Hence, we can assume that $C$ contains vertices of $C_f$ {\em and} vertices in the interior of $C_f$.

Consider a maximal path $P$ in $C$ all of whose edges belong to $C_f$. Let $u$ and $v$ be the end-vertices of $P$. Let $u'$ and $v'$ be the vertices adjacent to $u$ and $v$ in $C$, respectively, and not belonging to $P$. By the maximality of $P$, we have that $u'$ and $v'$ are in the interior of $C_f$. It might be the case that $u=v$ or that $u'=v'$ (however the two equalities do not hold simultaneously). Assume w.l.o.g. that $u'$, $u$, $v$, and $v'$ appear in this clockwise order along $C$. Denote by $b$ any node of ${\cal CBC}(H^*_{i+1,k},\psi_k)$ such that $u'$ belongs to $H^*_{i+1,k}(b)$.

Suppose first that algorithm cycle-breaker inserted vertex $r(u',b)$ into $B_f$ as the rightmost neighbor of $u'$ in $C_f$ from $b$. Assume also that $v'\neq u'$. By the assumption on the clockwise order of the vertices along $C$ and since all the vertices of $P$ are white, we have that edge $(v,v')$ crosses edge $(u',r(u',b))$, a contradiction to the planarity of $H^*_{i+1,k}(b)$ (see Fig.~\ref{fig:correctness}(a)). Assume next that $v'=u'$. Consider the edge $(u',u'')$ that follows $(u',u)$ in the clockwise order of the edges incident to $u'$. (Note that $u'' \neq v$, otherwise $C$ would be an empty triangle.) If $u''$ belongs to $C_f$, then it either belongs to $P$ or it does not. In the former case, a cycle $C'$ can be obtained from $C$ by replacing path $(u',u,u'')$ with edge $(u',u'')$; since $H^*_{i,j}$ is internally-triangulated, then $(u,u',u'')$ is a cycle delimiting a face of $H^*_{i,j}$, hence if $C$ contains a vertex $x$ in its interior, then $C'$ contains $x$ in its interior as well, thus contradicting the minimality of $C$ (see Fig.~\ref{fig:correctness}(b)). In the latter case, edge $(u',u'')$ crosses $C$, thus contradicting the planarity of $H^*_{i,j}$. We can hence assume that $u''$ belongs to $H^*_{i+1,k}$. Let $b'$ be the node in ${\cal CBC}(H^*_{i+1,k},\psi_k)$ such that $H^*_{i+1,k}(b')$ contains edge $(u',u'')$ (it is possible that $b'=b$). Observe that, by the planarity of $H^*_{i,j}$, graph $H^*_{i+1,k}(b')-u'$ lies entirely in the interior of $C$. This implies that $u$ is the rightmost neighbor $r(u',b')$ of $u'$ in $C_f$ from $b'$. Thus, if algorithm cycle-breaker inserted $r(u',b')$ into $B_f$, then we immediately get a contradiction to the fact that $u$ is white. Otherwise, vertex $u''$ is white, and algorithm cycle-breaker inserted into $B_f$ the rightmost neighbor $r(u'',b')$ of $u''$ in $C_f$ from $b'$. Since every vertex of $P$ is white, then $r(u'',b')$ does not belong to $P$, hence edge $(u'',r(u'',b'))$ crosses edge $(u',u)$ or edge $(u',v)$, a contradiction to the planarity of $H^*_{i,j}$ (see Fig.~\ref{fig:correctness}(c)).

\begin{figure}[tb]
\begin{center}
\begin{tabular}{c c c c}
\mbox{\includegraphics[scale=0.48]{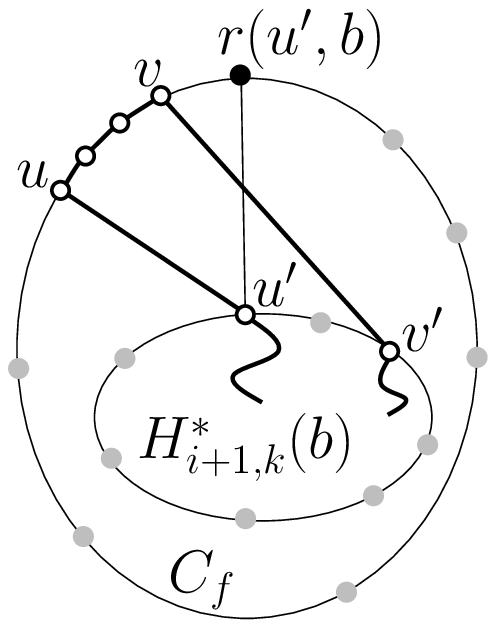}} \hspace{2mm} &
\mbox{\includegraphics[scale=0.48]{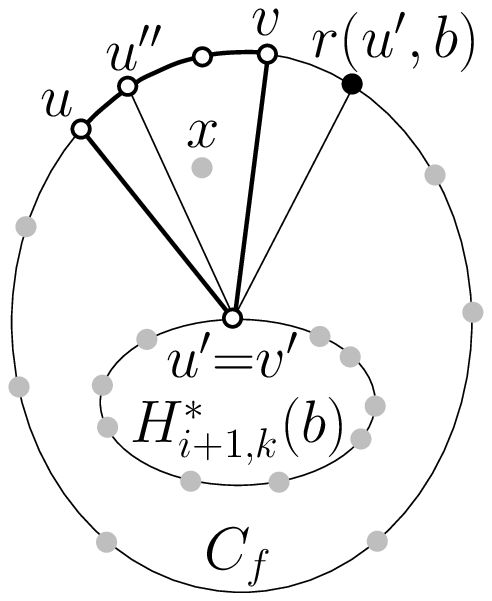}} \hspace{2mm} &
\mbox{\includegraphics[scale=0.48]{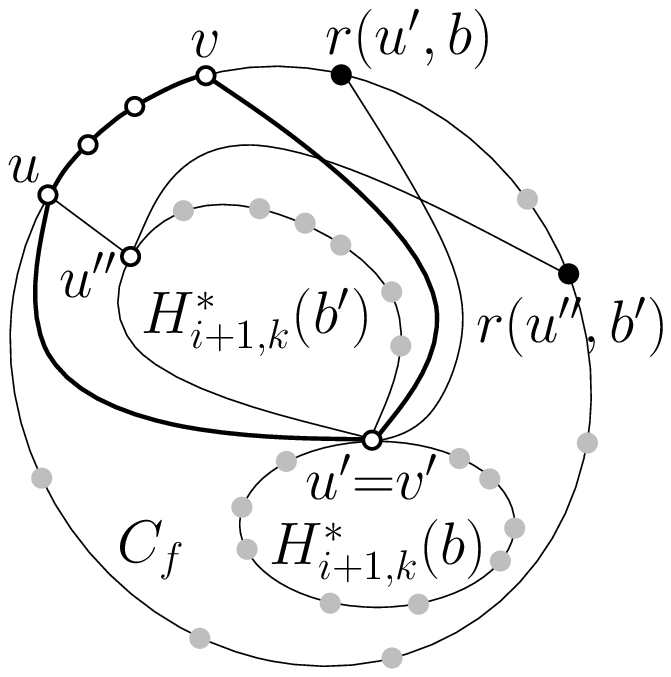}} \hspace{1mm} &
\mbox{\includegraphics[scale=0.48]{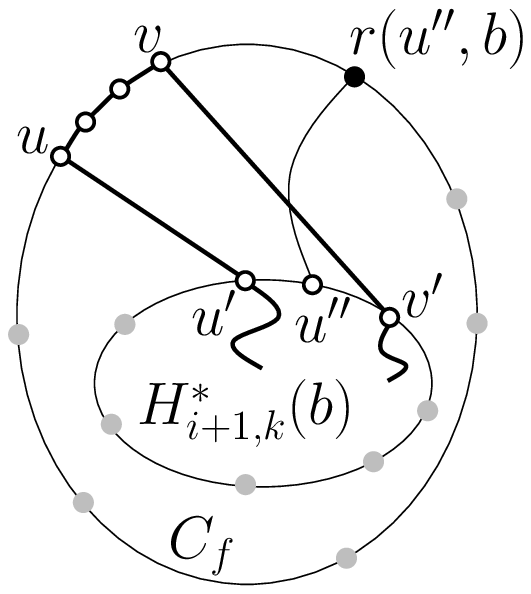}}\\
(a) \hspace{2mm}& (b) \hspace{2mm} & (c) \hspace{1mm} & (d)
\end{tabular}
\caption{Contradiction to the existence of a cycle $C$ (shown by thick lines) in $H^*_{i,j}[W^*_{i,j}(f)]$ containing a vertex $x$ in its interior. The color of the gray vertices is not important for the proof. Figures (a), (b), and (c) illustrate the case in which algorithm cycle-breaker inserted a vertex $r(u',b)$ into $B_f$, and in which $v'\neq u'$ (a), $v'= u'$ and $u'' \in V(P)$ (b), and $v'= u'$ and $u''\in V(H^*_{i+1,k})$ (c). Figure (d) illustrates the case in which algorithm cycle-breaker did not insert a vertex $r(u',b)$ into $B_f$ and $v'\neq u',u''$.}
\label{fig:correctness}
\end{center}
\end{figure}

Suppose next that algorithm cycle-breaker did not insert vertex $r(u',b)$ into $B_f$ as the rightmost neighbor of $u'$ in $C_f$ from $b$. Denote by $u''$ the vertex that follows $u'$ in the clockwise order of the vertices along the outer face of $H^*_{i+1,k}(b)$. Assume first that $v'=u''$. Then, edge $(u',v')$ splits $C$ into two cycles $C'$ and $C''$, with $|V(C')|\subset |V(C)|$ and $|V(C'')|\subset |V(C)|$, one of which contains $x$ in its interior, thus contradicting the minimality of $C$. Hence, we can assume that $v' \neq u''$. In the case in which $v' = u'$, a contradiction can be derived with {\em exactly} the same arguments as in the case in which algorithm cycle-breaker inserted vertex $r(u',b)$ into $B_f$. Hence, we can assume that $v' \neq u'$. Since algorithm cycle-breaker did not insert vertex $r(u',b)$ into $B_f$, it follows that $u''$ is white and algorithm cycle-breaker inserted into $B_f$ the rightmost neighbor $r(u'',b)$ of $u''$ in $C_f$ from $b$. Since every vertex of $P$ is white, it follows that $r(u'',b)$ does not belong to $P$. Hence, edge $(u'',r(u'',b))$ crosses edge $(u,u')$ or edge $(v,v')$, a contradiction to the planarity of $H^*_{i+1,k}(b)$ (see Fig.~\ref{fig:correctness}(d)).

{\em We prove that $\psi$ satisfies Condition (2) of Lemma~\ref{th:inductive}}. Consider any ${\cal BU}$-node $b$ in the contracted block-cutvertex tree ${\cal CBC}(H^*_{i,j},\psi)$. Denote by $H_{i,j}(b)$ the outerplane graph induced by the vertices incident to the outer face of $H^*_{i,j}(b)$ or, equivalently, the subgraph of $H_{i,j}$ induced by the vertices in $H^*_{i,j}(b)$.

We distinguish three cases. In {\em Case A}, graph $H_{i,j}(b)$ contains at least one non-trivial non-empty internal face; in {\em Case B}, all the faces of $H_{i,j}(b)$ are either trivial or empty, and there exists a vertex $v^*_{1,2}$ incident to two trivial faces $f_1$ and $f_2$ of $H_{i,j}(b)$; finally, in {\em Case C}, all the faces of $H_{i,j}(b)$ are either trivial or empty, and there exists no vertex incident to two trivial faces and of $H_{i,j}(b)$. We prove that, in Cases~A and~B, Condition (2a) of Lemma~\ref{th:inductive} is satisfied, while in Case~C, Condition (2b) of Lemma~\ref{th:inductive} is satisfied.

In all cases, the surplus $s(H^*_{i,j}(b),\psi)$ is the sum of the surpluses $s(H^*_{i+1,k},\psi)$ of the connected components $H^*_{i+1,k}$ of $G^*_{i+1}$ inside the internal faces of $H_{i,j}(b)$, plus the number $|W^*_{i,j} \cap V(H_{i,j}(b))|$ of white vertices in $H_{i,j}(b)$, minus the number $|B^*_{i,j} \cap V(H_{i,j}(b))|$ of black vertices in $H_{i,j}(b)$, which is equal to $|\bigcup_{f} B_f|$. Denote by $n_t$ the number of trivial faces of $H_{i,j}(b)$ and by $n_n$ the number of non-trivial non-empty internal faces of $H_{i,j}(b)$.

\begin{itemize}
\item We discuss {\bf Case A}. First, the number of vertices inserted in $\bigcup_{f} B_f$ by algorithm cycle-breaker when looking at trivial faces of $H_{i,j}(b)$ is at most $n_t$, as for every trivial face $f$ of $H_{i,j}(b)$ at most one vertex is inserted into $B_f$. Also, the sum of the surpluses $s(H^*_{i+1,k},\psi)$ of the connected components $H^*_{i+1,k}$ of $G^*_{i+1}$ inside trivial faces of $H_{i,j}(b)$ is $n_t$, given that each connected component $H^*_{i+1,k}$ inside a trivial face is either a single white vertex, or it is such that all the vertices incident to the outer face of $H^*_{i+1,k}$ are black (hence by induction $s(H^*_{i+1,k},\psi)\geq |W^*_{i+1,k} \cap V(H_{i+1,k})|+1 = 1$).

    In the following, we prove that, for every non-trivial non-empty internal face $f$ of $H_{i,j}(b)$ containing a connected component $H^*_{i+1,k}$ of $G^*_{i+1}$ in its interior, algorithm cycle-breaker inserts into $B_f$ at most $s(H^*_{i+1,k},\psi)-1$ vertices. The claim implies that Condition (2a) of Lemma~\ref{th:inductive} is satisfied by $H^*_{i,j}(b)$. In fact, (1) the sum of the surpluses $s(H^*_{i+1,k},\psi)$ of the connected components $H^*_{i+1,k}$ of $G^*_{i+1}$ inside the internal faces of $H_{i,j}(b)$ is $n_t + \sum_f s(H^*_{i+1,k},\psi)$ (where the sum is over each connected component $H^*_{i+1,k}$ inside a non-trivial non-empty internal face $f$ of $H_{i,j}(b)$), (2) the number of white vertices in $H_{i,j}(b)$ is $|W^*_{i,j} \cap V(H_{i,j}(b))|$, and (3) the number of black vertices in $H_{i,j}(b)$ is at most $n_t + \sum_f (s(H^*_{i+1,k},\psi)-1)$ (where the sum is over each connected component $H^*_{i+1,k}$ inside a non-trivial non-empty internal face $f$ of $H_{i,j}(b)$). Hence, $s(H^*_{i,j}(b),\psi)\geq n_t - n_t + |W^*_{i,j} \cap V(H_{i,j}(b))| + n_n$. By the assumption of Case A, we have $n_n\geq 1$, and Condition (2a) of Lemma~\ref{th:inductive} follows.

    Consider any non-trivial non-empty internal face $f$ of $H_{i,j}(b)$ containing a connected component $H^*_{i+1,k}$ of $G^*_{i+1}$ in its interior. We consider the ${\cal BU}$-nodes of the contracted block-cutvertex tree ${\cal CBC}(H^*_{i+1,k},\psi)$ of $H^*_{i+1,k}$ one at a time. Denote by $n_{bu}$ the number of ${\cal BU}$-nodes in ${\cal CBC}(H^*_{i+1,k},\psi)$ and by $b_1,b_2,\dots,b_{n_{bu}}$ the ${\cal BU}$-nodes of ${\cal CBC}(H^*_{i+1,k},\psi)$ in any order.

    We prove that, when algorithm cycle-breaker deals with ${\cal BU}$-node $b_l$, for any $1\leq l\leq n_{bu}$, it inserts into $B_f$ a number of vertices which is at most $s(H^*_{i+1,k}(b_l),\psi)-1$. Namely, if $s(H^*_{i+1,k}(b_l),\psi)\geq |W^*_{i,j} \cap V(H_{i+1,k}(b_l))|+1$, then it suffices to observe that, for each white vertex incident to the outer face of $H^*_{i+1,k}(b_l)$, at most one black vertex is inserted into $B_f$; further, if $s(H^*_{i+1,k}(b_l),\psi)= |W^*_{i,j} \cap V(H_{i+1,k}(b_l))|$ and there exists an edge $e$ incident to the outer face of $H^*_{i+1,k}(b_l)$ whose both end-vertices are white, then, for each white vertex incident to the outer face of $H^*_{i+1,k}(b_l)$, at most one black vertex is inserted into $B_f$ with the exception of one of the end-vertices of $e$, for which no black vertex is inserted into $B_f$. Hence, the number of vertices inserted into $B_f$ by algorithm cycle-breaker is at most $\sum_{l=1}^{n_{bu}} (s(H^*_{i+1,k}(b_l),\psi)-1)=\sum_{l=1}^{n_{bu}} s(H^*_{i+1,k}(b_l),\psi)-b_{n_{bu}}$.

    It remains to prove that $\sum_{l=1}^{n_{bu}} s(H^*_{i+1,k}(b_l),\psi)= s(H^*_{i+1,k},\psi)+b_{n_{bu}}-1$, which is done as follows. (Roughly speaking, if $b_{n_{bu}}>1$, then $\sum_{l=1}^{n_{bu}} s(H^*_{i+1,k}(b_l),\psi)>s(H^*_{i+1,k},\psi)$ holds because white cutvertices in $H^*_{i+1,k}$ belong to more than one graph $H^*_{i+1,k}(b_l)$, hence they contribute more than $1$ to $\sum_{l=1}^{n_{bu}} s(H^*_{i+1,k}(b_l),\psi)$, while they contribute exactly $1$ to $s(H^*_{i+1,k},\psi)$). Root ${\cal CBC}(H^*_{i+1,k},\psi)$ at any ${\cal BU}$-node, and orient all its edges towards the root. We now {\em assign} each white cutvertex $c_x$ in $H^*_{i+1,k}$ to the only ${\cal BU}$-node $b_l$ such that edge $(c_x,b_l)$ is oriented from $c_x$ to $b_l$. Such an assignment corresponds to consider $b_l$ as the only ${\cal BU}$-node in which $c_x$ is counted in order to relate $s(H^*_{i+1,k},\psi)$ to $\sum_{l=1}^{n_{bu}} s(H^*_{i+1,k}(b_l),\psi)$. Now the difference $\sum_{l=1}^{n_{bu}} s(H^*_{i+1,k}(b_l),\psi) -  s(H^*_{i+1,k},\psi)$ is equal to the number of pairs $(c_x,b_l)$ such that $c_x$ is not assigned to $b_l$, that is, the number of edges $(c_x,b_l)$ that are oriented from $b_l$ to $c_x$. For each ${\cal BU}$-node $b_l$ of ${\cal CBC}(H^*_{i+1,k},\psi)$, there is one such an edge, except for the root for which there is no such an edge. Hence, we get that $\sum_{l=1}^{n_{bu}} s(H^*_{i+1,k}(b_l),\psi)= s(H^*_{i+1,k},\psi)+b_{n_{bu}}-1$.

\item We now discuss {\bf Case B}. First, the sum of the surpluses $s(H^*_{i+1,k},\psi)$ of the connected components $H^*_{i+1,k}$ of $G^*_{i+1}$ inside the internal faces of $H_{i,j}(b)$ is equal to $n_t$, given that each connected component $H^*_{i+1,k}$ is either a single white vertex, or it is such that all the vertices incident to the outer face of $H^*_{i+1,k}$ are black (hence by induction $s(H^*_{i+1,k}(b),\psi)\geq |W^*_{i+1,k} \cap V(H_{i+1,k}(b))|+1 = 1$).

    Second, we prove that algorithm cycle-breaker defines $B_{f_1}=B_{f_2}=\{v^*_{1,2}\}$ for two trivial faces $f_1$ and $f_2$ of $H_{i,j}(b)$ sharing a vertex $v^*_{1,2}$. Suppose the contrary, for a contradiction. Two trivial faces $f_1$ and $f_2$ of $H_{i,j}(b)$ sharing a vertex $v^*_{1,2}$ exist by the assumption of Case B. Hence, algorithm cycle-breaker does not define $B_{f_1}=B_{f_2}=\{v^*_{1,2}\}$ only if it does define $B_{f_1}=B_{f_3}=\{v^*_{1,3}\}$, for some vertex $v^*_{1,3}$ incident to $f_1$ and to a trivial face $f_3$ of $H_{i,j}$ (possibly after swapping the labels of $f_1$ and $f_2$). If $f_3$ is a face of $H_{i,j}(b)$, we immediately get a contradiction. If $f_3$ is not a face of $H_{i,j}(b)$, then we get a contradiction since any vertex that is incident to an internal face of $H_{i,j}(b)$ and to an internal face of $H_{i,j}$ not in $H_{i,j}(b)$ is white, by definition of contracted block-cutvertex tree, hence it is not in $\bigcup_{f} B_f$.

    Thus, $|B^*_{i,j} \cap V(H_{i,j}(b))|=|\bigcup_{f} B_f|<n_t$. In fact, each trivial face contributes with at most one vertex to $\bigcup_{f} B_f$ and at least two trivial faces of $H_{i,j}(b)$ contribute with a total of one vertex to $\bigcup_{f} B_f$.

    Hence, $s(H^*_{i,j}(b),\psi) \geq n_t + |W^*_{i,j} \cap V(H_{i,j}(b))| - (n_t -1) = |W^*_{i,j} \cap V(H_{i,j}(b))| +1$, thus Condition (2a) of Lemma~\ref{th:inductive} is satisfied.

\item We now discuss {\bf Case C}. As in Case B, the sum of the surpluses of the connected components $H^*_{i+1,k}$ of $G^*_{i+1}$ inside the internal faces of $H_{i,j}(b)$ is equal to $n_t$.

    Second, $|B^*_{i,j} \cap V(H_{i,j}(b))|=|\bigcup_{f} B_f|= n_t$, as each trivial face contributes with exactly one vertex to $\bigcup_{f} B_f$. (Notice that, since no two trivial faces share a vertex, then no two trivial faces contribute with the same vertex to $\bigcup_{f} B_f$.)

    Hence, $s(H^*_{i,j}(b),\psi) = n_t + |W^*_{i,j} \cap V(H_{i,j}(b))| - n_t = |W^*_{i,j} \cap V(H_{i,j}(b))|$. Thus, in order to prove that Condition (2b) of Lemma~\ref{th:inductive} is satisfied, it remains to prove that there exists an edge incident to the outer face of $H^*_{i,j}(b)$ whose end-vertices belong to $W^*_{i,j}$. We will in fact prove that there exists an edge incident to the outer face of $H^*_{i,j}(b)$ whose end-vertices belong to $W^*_{i,j}$ in {\em every} $2$-connected component $D^*_{i,j}(b)$ of $H^*_{i,j}(b)$. Denote by $D_{i,j}(b)$ the outerplane graph induced by the vertices incident to the outer face of $D^*_{i,j}(b)$. Refer to Fig.~\ref{fig:casec}.

    \begin{figure}[tb]
        \begin{center}
        \mbox{\includegraphics[scale=0.6]{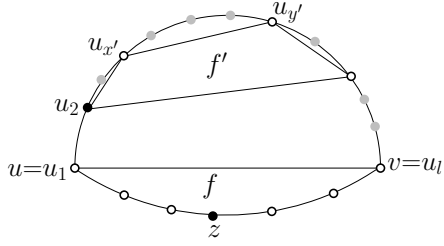}}
        \caption{Illustration for the proof that there exists an edge incident to the outer face of $D^*_{i,j}(b)$ whose end-vertices belong to $W^*_{i,j}$. }
        \label{fig:casec}
        \end{center}
    \end{figure}

    Suppose, for a contradiction, that no edge incident to the outer face of $D^*_{i,j}(b)$ has both its end-vertices in $W^*_{i,j}$. If all the faces of $D_{i,j}(b)$ are empty, then every edge incident to the outer face of $D^*_{i,j}(b)$ has both its end-vertices in $W^*_{i,j}$, thus obtaining a contradiction. Assume next that $D_{i,j}(b)$ has at least one trivial face $f$. Since exactly one of the vertices incident to $f$, say vertex $z$, belongs to $\bigcup_{f} B_f$ (as otherwise two trivial faces of $D_{i,j}(b)$ would exist sharing a vertex), it follows that at least one of the edges delimiting $f$ has both its end-vertices in $W^*_{i,j}$. Denote by $(u,v)$ such an edge and assume, w.l.o.g., that $u$, $v$, and $z$ appear in this clockwise order along the cycle delimiting $f$. If $(u,v)$ is incident to the outer face of $D_{i,j}(b)$, we immediately have a contradiction. Otherwise, $(u,v)$ is an internal edge of $D_{i,j}(b)$. Denote by $u=u_1,u_2,\dots,u_l=v$ the clockwise order of the vertices along the cycle delimiting the outer face of $D_{i,j}(b)$ from vertex $u$ to vertex $v$. We assume w.l.o.g. that $(u,v)$ is {\em maximal}, that is, there is no edge $(u_x,u_y)\neq (u,v)$ such that (1) $1\leq x < y \leq l$, (2) $u_x$ and $u_y$ are both white, and (3) there exists a trivial face $f'$ of $D_{i,j}(b)$ that is incident to edge $(u_x,u_y)$ and that is internal to cycle $C_{1,2}=(u_1,u_2,\dots,u_x,u_y,u_{y+1},\dots,u_l)$. Then, consider vertex $u_2$. If it is white, then we have a contradiction, as edge $(u_1,u_2)$ is incident to the outer face of $D_{i,j}(b)$. Otherwise, $u_2$ is black. Then, there exists a trivial face $f'$ such that $B_{f'}=\{u_2\}$. Since no vertex incident to $f'$ and different from $u_2$ belongs to $\bigcup_{f} B_f$ (as otherwise two trivial faces of $D_{i,j}(b)$ would exist sharing a vertex), it follows that at least one of the edges delimiting $f'$, say $e'=(u_{x'},u_{y'})$, has both its end-vertices in $W^*_{i,j}$. By planarity, the end-vertices of $e'$ are among $u_1,u_2,\dots,u_l$. Further, they are different from $u_1$ and $u_l$, as otherwise $f$ and $f'$ would share a vertex. Hence, $f'$ is internal to cycle $(u_1,u_2,\dots,u_{x'},u_{y'},u_{y'+1},\dots,u_l)$, thus contradicting the maximality of $(u,v)$.
\end{itemize}

This concludes the proof of the lemma.
\end{proof}


\section{Proof of Theorem~\ref{th:plane3trees-main}} \label{se:3trees}

In this section we prove Theorem~\ref{th:plane3trees-main}. It suffices to
prove Theorem~\ref{th:plane3trees-main} for an $n$-vertex {\em maximal} plane
graph $G_1$ and an $n$-vertex ({\em maximal}) plane $3$-tree $G_2$. In fact, if $G_1$ and $G_2$ are not maximal, then they can be augmented to an $n$-vertex
maximal plane graph $G'_1$ and an $n$-vertex plane $3$-tree $G'_2$, respectively; the latter augmentation can be always performed, as proved in~\cite{kv-nppt-12}. Then, a \sefenomap can be constructed for $G'_1$ and $G'_2$, and finally
the edges not in $G_1$ and $G_2$ can be removed, thus obtaining a \sefenomap of $G_1$
and $G_2$. In the following we assume that $G_1$ and $G_2$ are an $n$-vertex
maximal plane graph and an $n$-vertex plane $3$-tree, respectively, for
some $n\geq 3$. Denote by $C_i=(u_i,v_i,z_i)$ the cycle delimiting the outer
face of $G_i$, for $i=1,2$, where vertices $u_i$, $v_i$, and $z_i$ appear in
this clockwise order along $C_i$.

Let $p_u$, $p_v$, and $p_z$ be three points in the plane. Let $s_{uv}$,
$s_{vz}$, and $s_{zu}$ be three curves connecting $p_u$ and $p_v$,
connecting $p_v$ and $p_z$, and connecting $p_z$ and $p_u$, respectively, that
do not intersect except at their common end-points. Let $\Delta_{uvz}$ be
the closed curve $s_{uv}\cup s_{vz}\cup s_{zu}$. Assume that $p_u$, $p_v$, and
$p_z$ appear in this clockwise order along $\Delta_{uvz}$. Denoting by
$\int(\Delta)$ the interior of a closed curve $\Delta$, let
$\cl(\Delta)=\int(\Delta) \cup \Delta$. Let $P$ be a set of $n-3 \geq 0$ points
in $\int(\Delta_{uvz})$ and let $R$ be a set of points on $\Delta_{uvz}$, where
$p_u,p_v,p_z\in R$. Let $S$ be a set of curves whose end-points are in
$R\cup P$ such that: (i) No two curves in $S$ intersect, except possibly
at common end-points, (ii) no two curves in $S$ connect the same pair of
points in $R\cup P$, (iii) each curve in $S$ is contained in
$\cl(\Delta_{uvz})$, (iv) any point in $R$, except possibly for $p_u$, $p_v$,
and $p_z$, has exactly one incident curve in $S$, and (v) no curve in
$S$ connects two points of $R$ both lying on $s_{uv}$, or both lying on
$s_{vz}$, or both lying on $s_{zu}$. See Fig.~\ref{fig:plane3treesdrawing}(a).
We show the following.

\begin{figure}[tb]
\begin{center}
\begin{tabular}{c c}
\mbox{\includegraphics[scale=0.4]{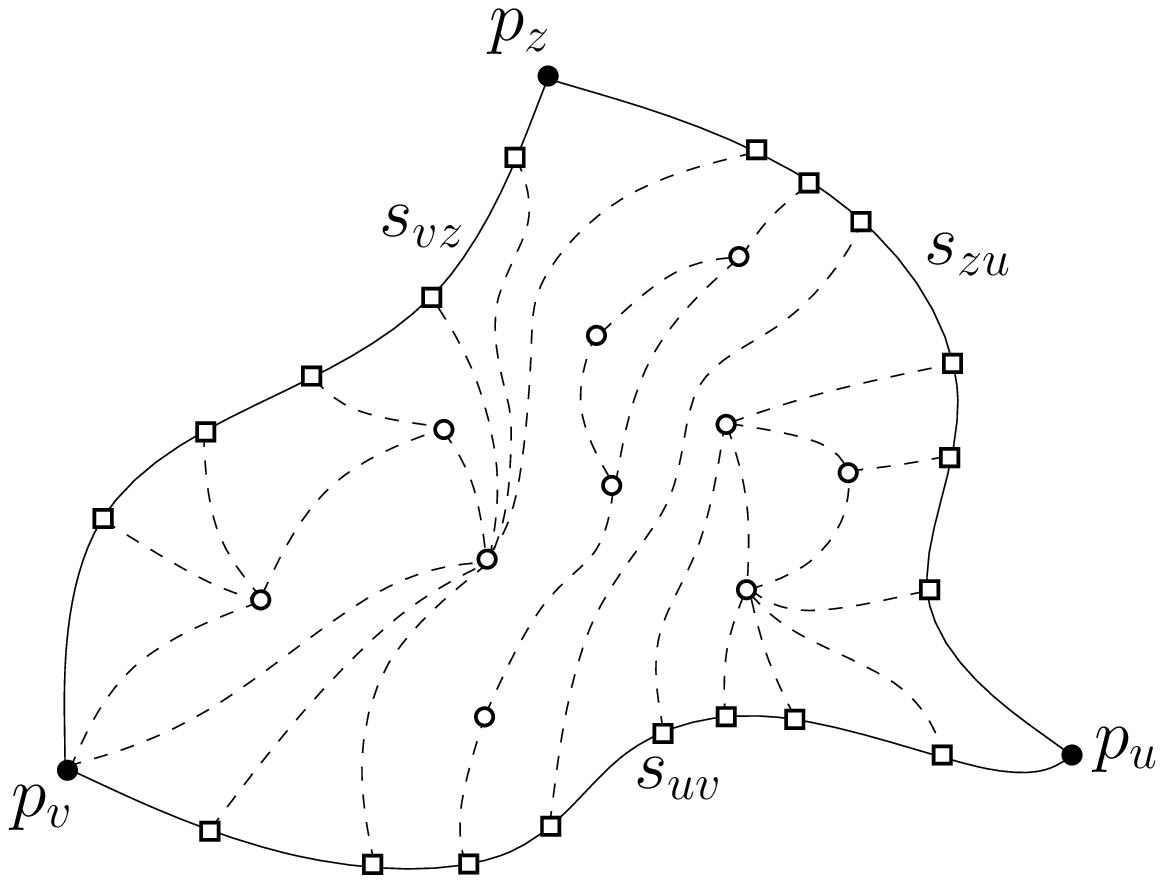}} \hspace{1mm} &
\mbox{\includegraphics[scale=0.4]{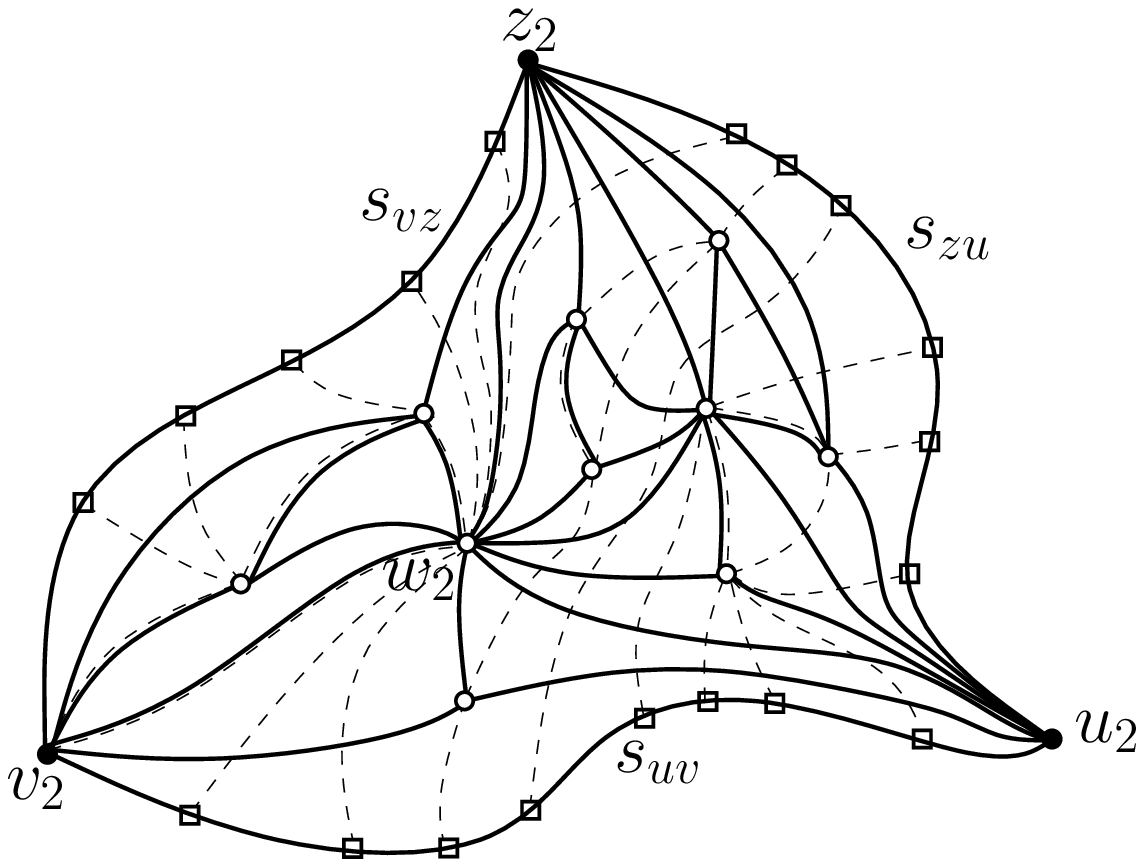}}\\
(a) \hspace{2mm}& (b)
\end{tabular}
\caption{(a) Setting for Lemma~\ref{le:splittingpoints}. White circles are points in $P$. White squares are points in $R$. Dashed curves are in $S$. Curves $s_{uv}$, $s_{vz}$, and $s_{zu}$ are solid thin curves. (b) A planar drawing of $G_2$ (solid thick lines) satisfying the properties of Lemma~\ref{le:splittingpoints}.}
\label{fig:plane3treesdrawing}
\end{center}
\end{figure}

\begin{lemma} \label{le:splittingpoints}
There exists a planar drawing $\Gamma_2$ of $G_2$ such that:
\begin{itemize}
\item[(a)] Vertices $u_2$, $v_2$, and $z_2$ are mapped to $p_u$, $p_v$, and $p_z$, respectively;
\item[(b)] edges $(u_2,v_2)$, $(v_2,z_2)$, and $(z_2,u_2)$ are represented by curves $s_{uv}$, $s_{vz}$, and $s_{zu}$, respectively;
\item[(c)] the internal vertices of $G_2$ are mapped to the points of $P$;
\item[(d)] each edge of $G_2$ that connects two points $p_1,p_2\in P \cup\{p_u, p_v, p_z\}$ such that there exists a curve $s\in S$ connecting $p_1$ and $p_2$ is represented by $s$ in $\Gamma_2$; and
\item[(e)] each edge $e$ of $G_2$ and each curve $s \in S$ such that $e$ is not represented by $s$ in $\Gamma_2$ cross at most once.
\end{itemize}
\end{lemma}

\begin{proof}
We prove the statement by induction on $n$. If $n=3$, then construct $\Gamma_2$ by mapping vertices $u_2$, $v_2$, and $z_2$ to $p_u$, $p_v$, and $p_z$ (thus satisfying Property (a)), respectively, and by mapping  edges $(u_2,v_2)$, $(v_2,z_2)$, and $(z_2,u_2)$ to curves $s_{uv}$, $s_{vz}$, and $s_{zu}$ (thus satisfying Property (b)). Property (c) is trivially satisfied since $G_2$ has no internal vertices and hence $P=\emptyset$. Property (d) is trivially satisfied since $G_2$ has no internal edges and hence $S=\emptyset$. Finally, Property (e) is trivially satisfied since $S=\emptyset$.

Suppose next that $n>3$. By the properties of plane $3$-trees, $G_2$ has an internal vertex $w_2$ that is connected to all of $u_2$, $v_2$, and $z_2$. Also, the subgraphs $G^{uv}_2$, $G^{vz}_2$, and $G^{zu}_2$ of $G_2$ induced by the vertices inside or on the border of cycles $C_2^{uv}=(u_2,v_2,w_2)$, $C_2^{vz}=(v_2,z_2,w_2)$, and $C_2^{zu}=(z_2,u_2,w_2)$, respectively, are plane $3$-trees with $n_{uv}$, $n_{vz}$, and $n_{zu}$ internal vertices, respectively, where $n_{uv}+n_{vz}+n_{zu}=n-4$.

We claim that there exists a point $p_w \in P$ and three curves $s_{uw}$, $s_{vw}$, and $s_{zw}$  connecting $p_u$ and $p_w$, connecting $p_v$ and $p_w$, and connecting $p_z$ and $p_w$, respectively, such that the following hold (see Fig.~\ref{fig:plane3treesclaim}):

\begin{figure}[htb]
\begin{center}
\mbox{\includegraphics[scale=0.48]{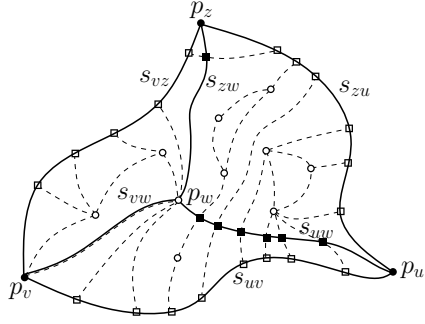}}
\caption{Illustration for the claim to prove Lemma~\ref{le:splittingpoints}. White circles represent points in $P$. White squares represent points in $R$. Points in $R_{uv}$, in $R_{vz}$, and in $R_{zu}$ are black squares. Dashed curves are in $S$. Curves $s_{uv}$, $s_{vz}$, $s_{zu}$, $s_{uw}$, $s_{vw}$, and $s_{zw}$ are solid thick curves. In this example $n_{uv}=1$, $n_{vz}=2$, and $n_{zu}=6$. }
\label{fig:plane3treesclaim}
\end{center}
\end{figure}

\begin{itemize}
\item[(P1)] $s_{uw}$, $s_{vw}$, and $s_{zw}$ do not intersect each other and do not intersect $s_{uv}$, $s_{vz}$, and $s_{zu}$, other than at common end-points;
\item[(P2)] if there exists a curve $s_a\in S$ connecting $p_u$ and $p_w$, then $s_{uw}$ coincides with $s_a$; if there exists a curve $s_b\in S$ connecting $p_v$ and $p_w$, then $s_{vw}$ coincides with $s_b$; if there exists a curve $s_c\in S$ connecting $p_z$ and $p_w$, then $s_{zw}$ coincides with $s_c$;
\item[(P3)] for any curve $s\in S$ that does not coincide with $s_{uw}$, curves $s$ and $s_{uw}$ cross at most once; for any curve $s\in S$ that does not coincide with $s_{vw}$, curves $s$ and $s_{vw}$ cross at most once; for any curve $s\in S$ that does not coincide with $s_{zw}$, curves $s$ and $s_{zw}$ cross at most once;
\item[(P4)] the closed curve $\Delta_{uvw}=s_{uv}\cup s_{uw}\cup s_{vw}$ contains in its interior a subset $P_{uv}$ of $P$ with $n_{uv}$ points; the closed curve $\Delta_{vzw}=s_{vz}\cup s_{vw}\cup s_{zw}$ contains in its interior a subset $P_{vz}$ of $P$ with $n_{vz}$ points; and the closed curve $\Delta_{zuw}=s_{zu}\cup s_{zw}\cup s_{uw}$ contains in its interior a subset $P_{zu}$ of $P$ with $n_{zu}$ points; and
\item[(P5)] if a curve $s\in S$ has both its end-points in $\cl(\Delta_{uvw})$, in $\cl(\Delta_{vzw})$, or in $\cl(\Delta_{zuw})$, then it is entirely contained in $\cl(\Delta_{uvw})$, in $\cl(\Delta_{vzw})$, or in $\cl(\Delta_{zuw})$, respectively.
\end{itemize}

We first prove that the claim implies the lemma, and we later prove the claim.

Suppose that the claim holds. Denote by $R_{uv}$ the set of points consisting of: (i) the points in $R$ lying on $s_{uv}$, (ii) the intersection points of $s_{uw}$ with the curves in $S$, if $s_{uw}$ does not coincide with any edge in $S$, and (iii) the intersection points of $s_{vw}$ with the curves in $S$, if $s_{vw}$ does not coincide with any edge in $S$. Analogously define $R_{vz}$ and $R_{zu}$. Let $S^+$ be the set of curves obtained by subdividing each curve $s\in S$ with its intersection points with $s_{uw}$, $s_{vw}$, and $s_{zw}$ (e.g., if $s$ connects points $q_1,q_2 \in R\cup P$ and it has an intersection $q_3$ with $s_{uw}$, then $S^+$ contains two curves connecting $q_1$ and $q_3$, and connecting $q_3$ and $q_2$, respectively). Observe that no curve in $s\in S^+$ has both its endpoints lying on $s_{uw}$, or both its endpoints lying on $s_{vw}$, or both its endpoints lying on $s_{zw}$, as otherwise the curve in $S$ of which $s$ is part would cross twice $s_{uw}$ or $s_{vw}$ or $s_{zw}$, respectively, contradicting Property P3 of the claim. Also, denote by $S_{uv}$, $S_{vz}$, and $S_{zu}$ the subsets of $S^+$ composed of the curves in $\cl(\Delta_{uvw})$, in $\cl(\Delta_{vzw})$, and in $\cl(\Delta_{zuw})$, respectively.

Apply induction three times. The first time to construct a drawing $\Gamma^{uv}_2$ of $G^{uv}_2$ (where the parameters $p_u$, $p_v$, $p_z$, $s_{uv}$, $s_{vz}$, $s_{zu}$, $\Delta_{uvz}$, $P$, $R$, and $S$ in the statement of Lemma~\ref{le:splittingpoints} are replaced with $p_u$, $p_v$, $p_w$, $s_{uv}$, $s_{uw}$, $s_{vw}$, $\Delta_{uvw}$, $P_{uv}$, $R_{uv}$, and $S_{uv}$, respectively), the second time to construct a drawing $\Gamma^{vz}_2$ of $G^{vz}_2$ (where the parameters $p_u$, $p_v$, $p_z$, $s_{uv}$, $s_{vz}$, $s_{zu}$, $\Delta_{uvz}$, $P$, $R$, and $S$ in the statement of Lemma~\ref{le:splittingpoints} are replaced with $p_v$, $p_z$, $p_w$, $s_{vz}$, $s_{vw}$, $s_{zw}$, $\Delta_{vzw}$, $P_{vz}$, $R_{vz}$, and $S_{vz}$, respectively), and the third time to construct a drawing $\Gamma^{zu}_2$ of $G^{zu}_2$ (where the parameters $p_u$, $p_v$, $p_z$, $s_{uv}$, $s_{vz}$, $s_{zu}$, $\Delta_{uvz}$, $P$, $R$, and $S$ in the statement of Lemma~\ref{le:splittingpoints} are replaced with $p_z$, $p_u$, $p_w$, $s_{zu}$, $s_{zw}$, $s_{uw}$, $\Delta_{zuw}$, $P_{zu}$, $R_{zu}$, and $S_{zu}$, respectively). Observe that induction can be applied since, by Properties P4 of the claim, the number of points in $P_{uv}$,  in $P_{uv}$, and in $P_{uv}$ is equal to the number of internal vertices of $G^{uv}_2$, of $G^{vz}_2$, and of $G^{zu}_2$, respectively.

Placing $\Gamma^{uv}_2$, $\Gamma^{vz}_2$, and $\Gamma^{zu}_2$ together results in a drawing $\Gamma_2$ of $G_2$; in particular, edge $(u,w)$ is represented by curve $s_{uw}$ both in $\Gamma^{uv}_2$ and in $\Gamma^{vz}_2$ (analogous statements hold for $(v,w)$ and $(z,w)$). Drawing $\Gamma_2$ is planar by induction and by Property P1 of the claim; also, $\Gamma_2$ satisfies Properties (a) and (b) of the lemma by induction; further, $\Gamma_2$ satisfies Property (c) of the lemma by induction and by definition of $p_w$; moreover, $\Gamma_2$ satisfies Property (d) of the lemma by induction and by Properties P2 and P5 of the claim; finally, $\Gamma_2$ satisfies Property (e) of the lemma by induction and by Properties P3 and P5 of the claim.

We now prove the claim. First, we {\em almost-triangulate} the interior of $\Delta_{uvw}$, that is, we add a maximal set of curves $S'$ to $S$ such that:  (i) No two curves in $S \cup S'$ intersect, except possibly at common end-points, (ii) no two curves in $S \cup S'$ connect the same pair of points in $R\cup P$, (iii) each curve in $S \cup S'$ is contained in $\cl(\Delta_{uvz})$, (iv) any point in $R$, except possibly for $p_u$, $p_v$, and $p_z$, has exactly one incident curve in $S \cup S'$, and (v) no curve in $S \cup S'$ connects two points of $R$ both lying on $s_{uv}$, or both lying on $s_{vz}$, or both lying on $s_{zu}$.

Second, we prove the claim by induction on $n_{uv}+n_{vz}+n_{zu}$. In the base case, $n_{uv}+n_{vz}+n_{zu}=0$. Then, let $p_w$ be the only point in $P$. Refer to Fig.~\ref{fig:plane3treesbasecase}.

\begin{figure}[htb]
\begin{center}
\begin{tabular}{c c c}
\mbox{\includegraphics[scale=0.4]{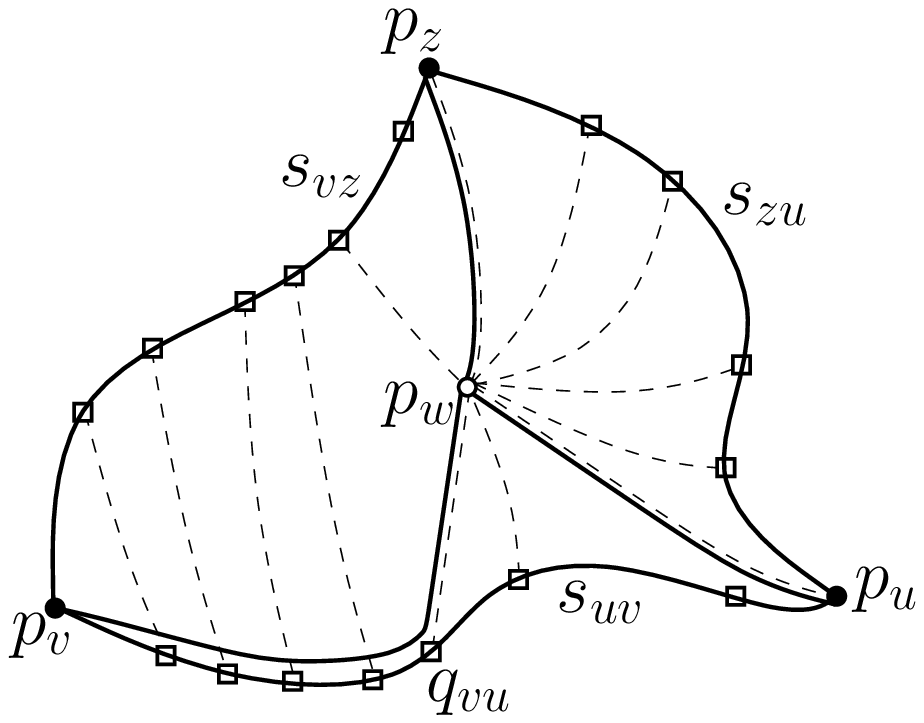}} \hspace{1mm} &
\mbox{\includegraphics[scale=0.4]{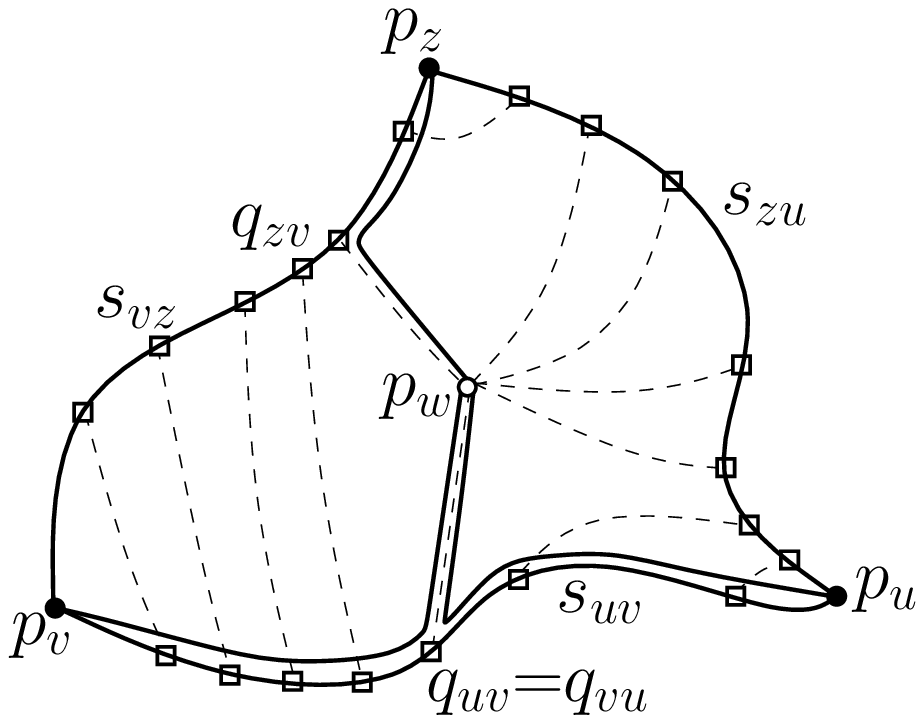}} \hspace{1mm} &
\mbox{\includegraphics[scale=0.4]{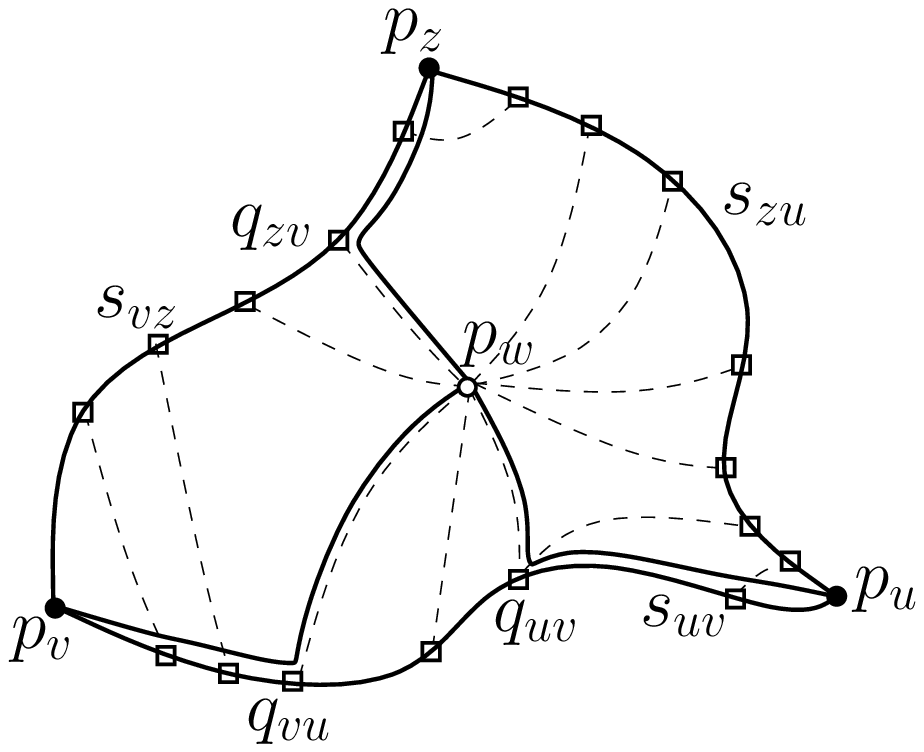}}
\end{tabular}
\caption{Three examples for the base case $n_{uv}+n_{vz}+n_{zu}=0$.}
\label{fig:plane3treesbasecase}
\end{center}
\end{figure}

We show how to draw $s_{uw}$.

\begin{itemize}
\item If a curve $s_a\in S$ exists connecting $p_u$ and $p_w$, then $s_{uw}$ coincides with $s_a$.

\item Suppose that no curve exists in $S$ connecting $p_u$ and $p_w$. If a curve $s_a\in S$ exists connecting $p_w$ and a point in $R$ on $s_{uv}$, then let $q_{uv}$ be the point in $R$ on $s_{uv}$ that is ``closest'' to $u$, i.e., no curve exists connecting $p_w$ and a point $q'_{uv}$ in $R$ such that $q'_{uv}$ lies on the part $s_{quv}$ of $s_{uv}$ between $q_{uv}$ and $u$. Draw $s_{uw}$ as a curve arbitrarily close to the curve $s_a \cup s_{quv}$.

\item Suppose that no curve exists in $S$ connecting $p_w$ and a point in $R$ on $s_{uv}$. If a curve $s_b\in S$ exists connecting $p_w$ and a point in $R$ on $s_{zu}$, then define $s_{quz}$ analogously to $s_{quv}$, and draw $s_{uw}$ as a curve arbitrarily close to the curve $s_b \cup s_{quz}$.

\item Finally, suppose that no curve exists in $S$ connecting $p_w$ with any point in $R$ on $s_{uv}$ or on $s_{zu}$. We claim that $p_w$ and $s_{uv}$ or $p_w$ and $s_{zu}$ are incident to a common face in $\cl(\Delta_{uvz})$. In fact, if $p_w$ and $s_{uv}$ are not incident to a common face in $\cl(\Delta_{uvz})$, then a curve $s_x$ exists in $S$ connecting a point in $R$ on $s_{vz}$ and a point in $R$ on $s_{zu}$ in such a way that $s_x$ ``separates'' $p_w$ from $s_{uv}$; analogously, if $p_w$ and $s_{zu}$ are not incident to a common face in $\cl(\Delta_{uvz})$, then a curve $s_y$ exists in $S$ connecting a point in $R$ on $s_{uv}$ and a point in $R$ on $s_{vz}$ in such a way that $s_y$ ``separates'' $p_w$ from $s_{zu}$; however this implies that $s_x$ and $s_y$ cross, contradicting the assumptions on $S$. Then, say that $p_w$ and $s_{uv}$  are incident to a common face $f$ in $\cl(\Delta_{uvz})$. Insert a dummy point $q_{uv}$ in $R$ on $s_{uv}$ incident to $f$ together with a curve $s_c$ connecting $p_w$ and $q_{uv}$ inside $f$. Define $s_{quv}$ as the part of $s_{uv}$ between $q_{uv}$ and $p_u$. Draw $s_{uw}$ as a curve arbitrarily close to the curve $s_c \cup s_{quv}$.
\end{itemize}

We draw $s_{vw}$ and $s_{zw}$ analogously to $s_{uw}$.

The constructed drawing is easily shown to satisfy Properties P1--P5 of the claim. In particular, Property P3 can be shown to be satisfied by $s_{uw}$ as follows (the proof for $s_{vw}$ and $s_{zw}$ is analogous). Consider any curve $s\in S$. If $s_{uw}$ coincides with a curve $s'\in S$, then $s_{uw}$ and $s$ do not cross, given that $s'$ and $s$ do not cross by the assumptions on $S$. Otherwise, $s_{uw}$ is arbitrarily close to $s_a \cup s_{quv}$, for some curve $s_a\in S$ incident to a point in $R$ on $s_{uv}$ or for some curve $s_a$ lying in the interior of a face $f$ of $\cl(\Delta_{uvz})$. Hence, if $s$ is incident to a point in $R$ on $s_{quv}$ different from $q_{uv}$, then $s_{uw}$ and $s$ cross exactly once arbitrarily close to $s_{quv}$. Otherwise, either $s_{uw}$ and $s$ share $p_w$ as common endpoint and do not cross at any other point, or $s_{uw}$ and $s$ do not share $p_w$ as common endpoint and do not cross at all, given that $s_a$ and $s$, as well as $s_{quv}$ and $s$, do not cross by the assumptions on $S$.

In the inductive case, $n_{uv}+n_{vz}+n_{zu}>0$. Suppose, w.l.o.g., that $n_{uv}>0$, the other cases being analogous.

We say that a point $p\in P$ is {\em close to $s_{uv}$} if the following condition holds. Let $G(P\cup R,S)$ be the plane graph whose vertices are the points in $P\cup R$ and whose edges are the curves in $S$. See Fig.~\ref{fig:closetouv}(a). Also, let $G(u,v,p)$ be the plane subgraph of $G(P\cup R,S)$ whose vertices are the points in $P\cup R$ and whose edges are: (i) The curves that compose $s_{uv}$, and (ii) every curve in $S$ that is incident to $p$ and to a point on $s_{u,v}$ (including $p_u$ and $p_v$). Then $p$ is close to $s_{uv}$ if: (a) It is incident to the same face of $G(P\cup R,S)$ a point on $s_{uv}$ is incident to, and (b) no point in $P$ lies inside any of the bounded faces of $G(u,v,p)$.

\begin{figure}[htb]
\begin{center}
\begin{tabular}{c c c}
\mbox{\includegraphics[scale=0.41]{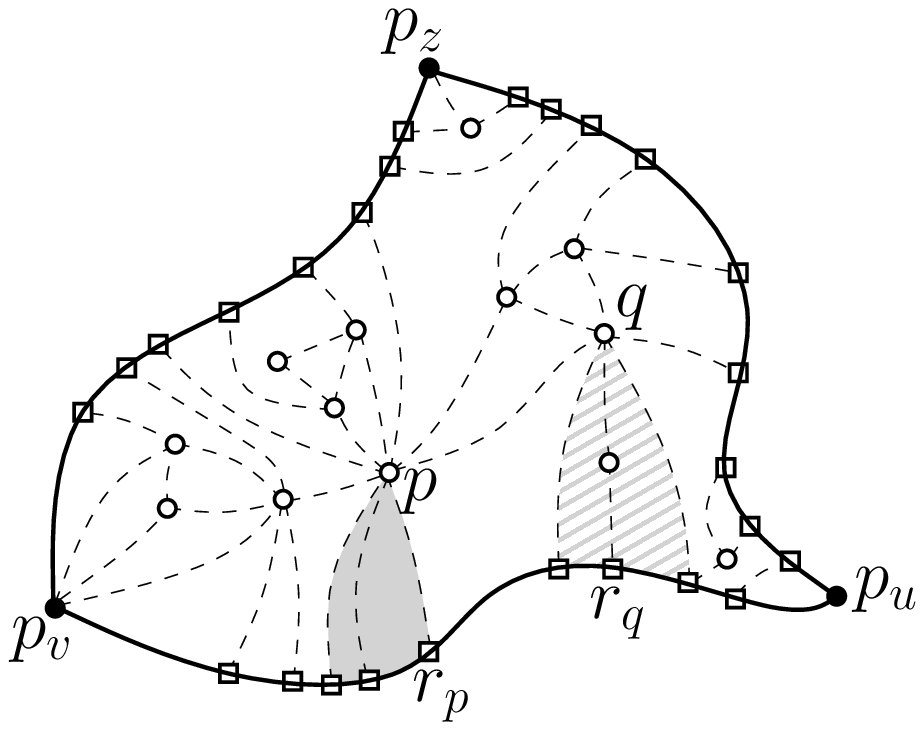}} \hspace{1mm} &
\mbox{\includegraphics[scale=0.41]{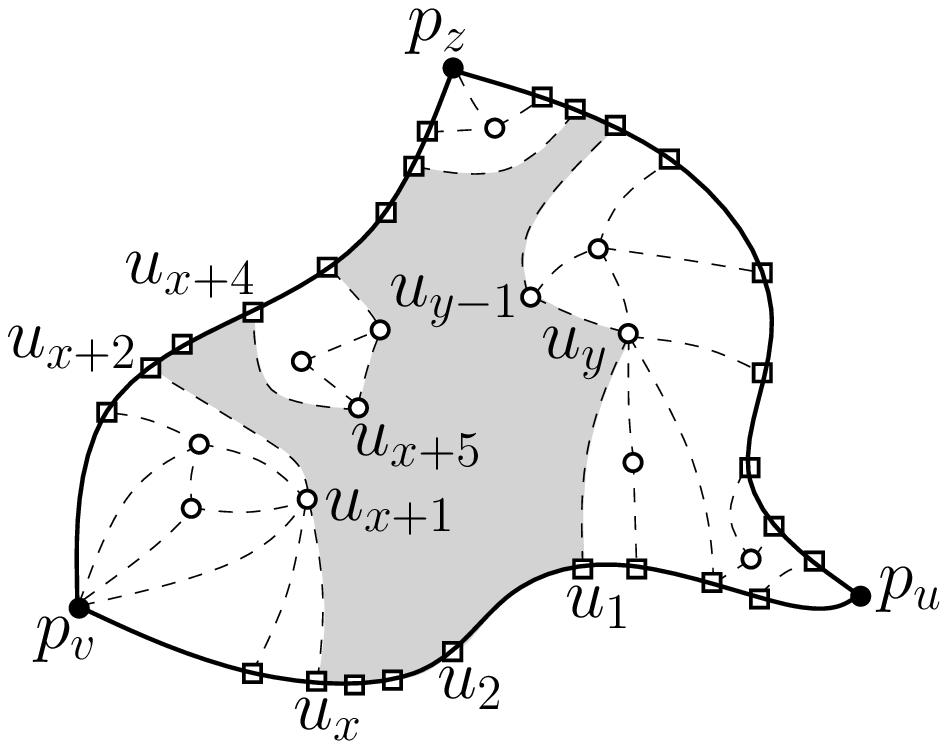}} \hspace{1mm} &
\mbox{\includegraphics[scale=0.41]{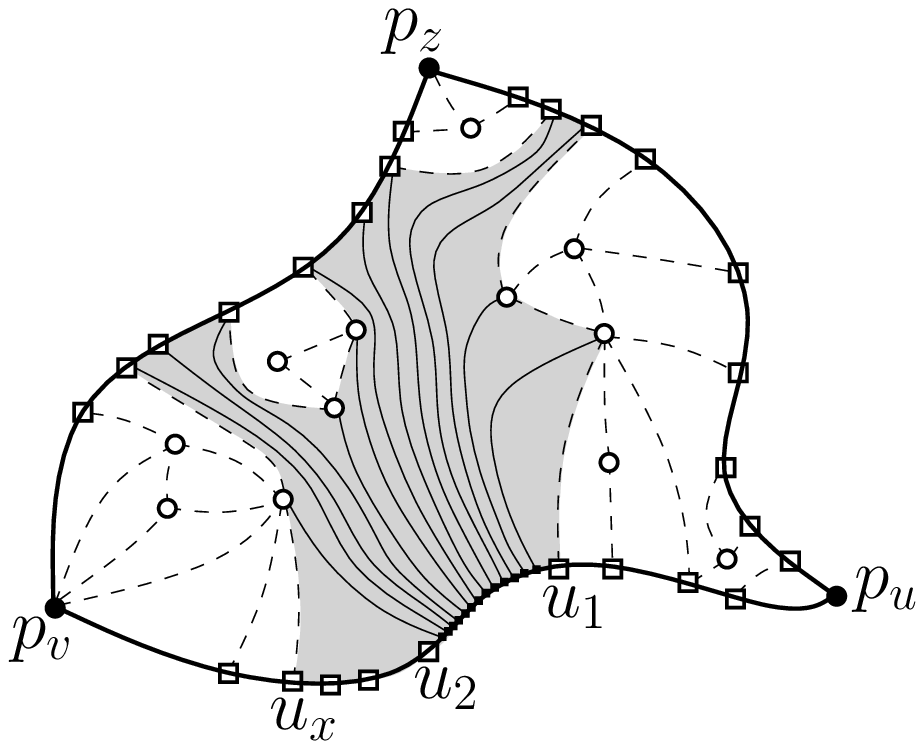}} \\
(a) \hspace{1mm} & (b) \hspace{1mm} & (c) \\
\mbox{\includegraphics[scale=0.41]{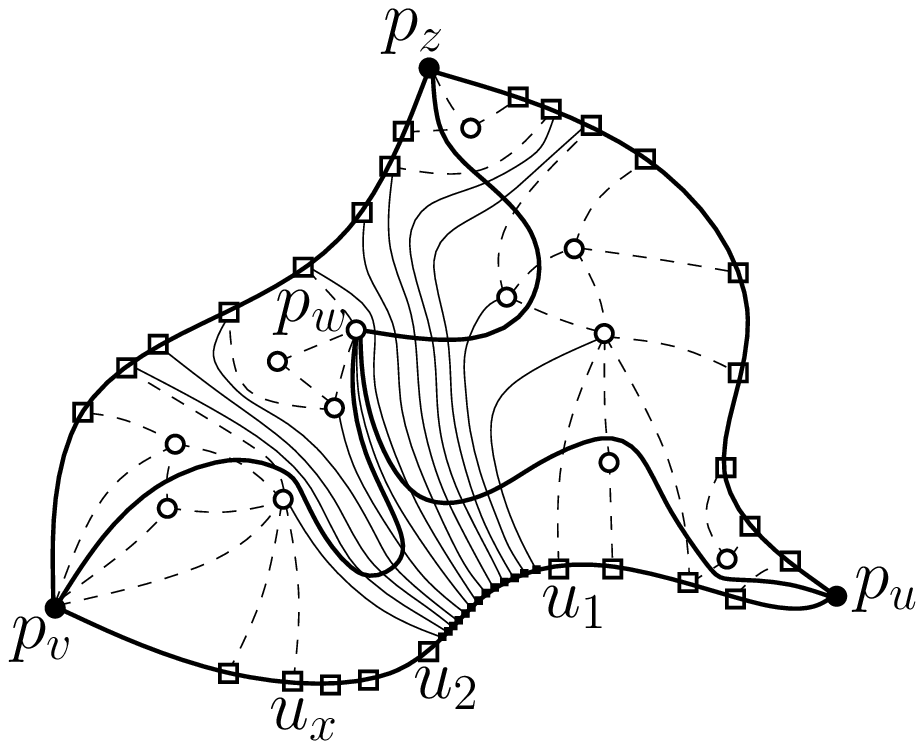}} \hspace{1mm} &
\mbox{\includegraphics[scale=0.41]{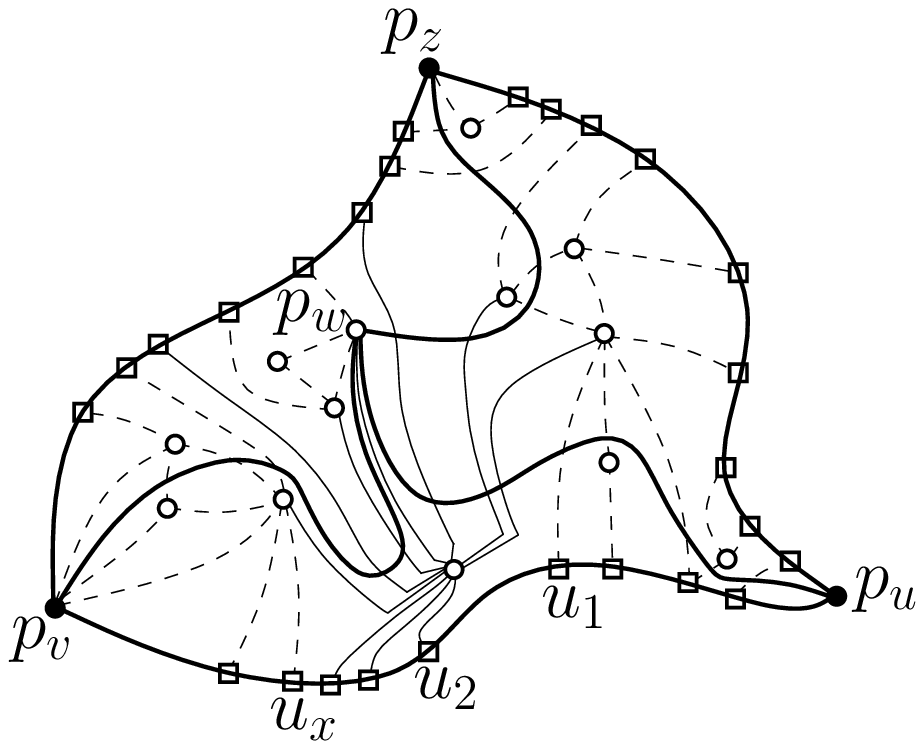}} \hspace{1mm} &
\mbox{\includegraphics[scale=0.41]{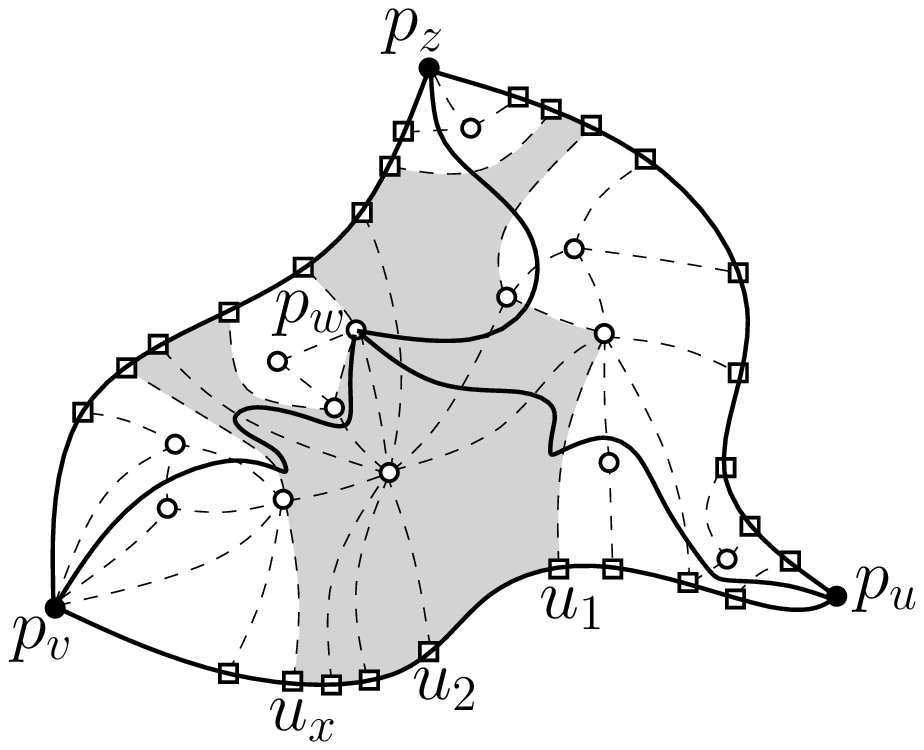}} \\
(d) \hspace{1mm} & (e) \hspace{1mm} & (f)
\end{tabular}
\caption{(a) Graph $G(P\cup R,S)$. Point $p$ is close to $s_{uv}$. In fact, $p$ and a point $r_p$ on $s_{uv}$ are incident to the same face of $G(P\cup R,S)$, and the bounded faces of $G(u,v,p)$, which are shown gray, do not contain any point in $P\cup R$. Point $q$ is not close to $s_{uv}$. In fact, although $q$ and a point $r_q$ on $s_{uv}$ are incident to the same face of $G(P\cup R,S)$, the only bounded face of $G(u,v,q)$, which is shown with gray stripes, contains a point in $P$. (b) Removal of $p$ from $G(P\cup R,S)$. The face $f$ in which $p$ used to lie is colored gray. (c) Insertion of dummy vertices $r'_{y-x},r'_{y-x-1},\dots,r'_1$ (small black squares) on edge $(u_1,u_2)$, and insertion of dummy edges $e_1,e_2,\dots,e_{y-x}$ (solid thin lines) inside $f$. (d) Inductively constructed drawings $s_{uw}$, $s_{vw}$, and $s_{zw}$. (e) Reintroduction of $p$ and its incident edges. (f) Restoration of the position of $p$ and of the drawing of its incident edges, while preserving the topology of the drawing. Only the part of $s_{uw}$, $s_{vw}$, and $s_{zw}$ in the interior of $f$ has to be modified for this sake.}
\label{fig:closetouv}
\end{center}
\end{figure}

A point close to $s_{u,v}$ always exists. Namely, consider any point $p\in P$ incident to the same face of $G(P\cup R,S)$ a point on $s_{uv}$ is incident to. If no point in $P$ lies inside any of the bounded faces of $G(u,v,p)$, then $p$ is the desired point. Otherwise, a set $F$ of isolated vertices lies inside a bounded face $f$ of $G(u,v,p)$. Among those points, there is a point $p'$ that is incident to the same face of $G(P\cup R,S)$ a point on $s_{uv}$ is incident to. Then, consider $G(u,v,p')$. If no point in $P$ lies inside any of the bounded faces of $G(u,v,p')$, then $p'$ is the desired point. Otherwise, a set $F'$ of isolated vertices lies inside a bounded face $f'$ of $G(u,v,p')$. However, $F'$ is a subset of $F$, thus the repetition of such an argument eventually leads to find a point close to $s_{uv}$.

Let $p$ be any point close to $s_{uv}$. Remove $p$ and its incident edges from $G(P\cup R,S)$. Let $f$ be the face of $G(P\cup R,S)$ in which $p$ used to lie and let $C_f$ be the cycle delimiting $f$. Since $p$ is close to $s_{uv}$, the vertices in $R$ lying on $s_{uv}$ appear consecutively along $C_f$. Denote by $u_1,u_2,\dots,u_y$ the clockwise order of the vertices along $C_f$, where $u_1,u_2,\dots,u_x$ are the vertices in $R$ on $s_{uv}$ in order from $p_u$ to $p_v$. See Fig.~\ref{fig:closetouv}(b). It holds $x\geq 2$, given that $p$ is incident to the same face of $G(P\cup R,S)$ a point on $s_{uv}$ is incident to, and given that every point in $R$ has exactly one incident curve in $S$.

Insert $y-x$ dummy vertices $r'_{y-x},r'_{y-x-1},\dots,r'_1$ in $R$ in this
order on edge $(u_1,u_2)$. Insert dummy edges $e_1,e_2,\dots,e_{y-x}$ inside $f$
from $u_{x+i}$ to $r'_i$, for every $1\leq i\leq y-x$. See
Fig.~\ref{fig:closetouv}(c). Inductively draw $s_{uw}$, $s_{vw}$, and $s_{zw}$
so that Properties P1--P5 of the claim are satisfied, where Property P4 ensures
that $\Delta_{uvw}$, $\Delta_{vzw}$, and $\Delta_{zuw}$ contain $n_{uv}-1$,
$n_{vz}$, and $n_{zu}$ points of $P\setminus\{p\}$ in their interior,
respectively. See Fig.~\ref{fig:closetouv}(d).

Introduce $p$ in a point arbitrarily close to edge $(u_1,u_2)$. Reintroduce the edges incident to $p$ as follows. Draw curves connecting $p$ and its neighbors among $u_1,u_2,\dots,u_x$ inside $f$ and arbitrarily close to $s_{uv}$. Also, for each neighbor $u_{x+i}$ of $p$ with $1\leq i\leq y-x$, draw a curve connecting $p$ and $u_{x+i}$ as composed of two curves, the first one arbitrarily close to $s_{uv}$, the second one coinciding with part of the edge $e_i$. Remove dummy vertices $r'_1,r'_2,\dots,r'_{y-x}$ and dummy edges $e_1,e_2,\dots,e_{y-x}$ from the drawing.  See Fig.~\ref{fig:closetouv}(e). Finally, restore the placement of $p$ and the drawing of its incident edges. In order to do so while maintaining the topology of the drawing (i.e. the number and order of the crossings along each edge), curves $s_{uw}$, $s_{vw}$, and $s_{zw}$ have to be modified in the interior of $f$.  See Fig.~\ref{fig:closetouv}(f).

The constructed drawing of $s_{uw}$, $s_{vw}$, and $s_{zw}$ satisfies Properties P1--P5 of the claim as shown in the following.

\begin{itemize}
\item Property P1 directly comes from induction.
\item We prove Property P2. Consider any curve $s\in S$. We prove that, if $s$ connects $p_u$ and $p_w$, then $s_{uw}$ coincides with $s$. Suppose first that $s$ is incident to $p$. By construction, $p\neq p_u,p_w$, hence in this case there is nothing to prove. Suppose next that $s$ is not incident to $p$. By induction, if $s$ connects $p_u$ and $p_w$, then $s_{uw}$ coincides with $s$ before restoring the position of $p$ and the drawing of its incident edges. Moreover, the only part of the drawing of $s_{uw}$ that can be possibly modified in order to restore the position of $p$ and the drawing of its incident edges is the one lying in the interior of $f$. However, if $s_{uw}$ coincides with $s$, then no part of it lies in the interior of $f$, hence $s_{uw}$ still coincides with $s$ after the modification. It can be analogously proved that if $s$ connects $p_v$ and $p_w$, or $p_z$ and $p_w$, then $s_{vw}$ or $s_{zw}$ coincides with $s$, respectively.
\item We prove Property P3. Consider any curve $s\in S$ and assume that $s$ does not coincide with $s_{uw}$. We prove that $s$ and $s_{uw}$ cross at most once. Since restoring the position of $p$ and the drawing of its incident edges does not alter the number of crossings between $s$ and $s_{uw}$, it suffices to prove that they cross at most once when $p$ and its incident edges are first reintroduced inside $f$. Suppose first that $s$ is not incident to $p$. Then, by induction $s$ and $s_{uw}$ cross at most once. Suppose next that $s$ is incident to $p$ and to a point $u_i$, for some $1\leq i\leq y$. If $1\leq i\leq x$, then $s$ is arbitrarily close to $s_{uv}$, hence it does not cross $s_{uw}$ at all. If $x+1\leq i\leq y$, then $s$ is composed of two curves, the first one arbitrarily close to $(u_1,u_2)$ (hence, such a curve does not cross $s_{uw}$ at all), the second one coinciding with part of edge $e_i$ (hence such a curve crosses $s_{uw}$ at most once, by induction). It can be analogously proved that if $s$ does not coincide with $s_{vw}$ or with $s_{zw}$, then $s$ and $s_{vw}$ or $s$ and $s_{zw}$ cross at most once, respectively.
\item In order to prove Property P4, it suffices to observe that by induction $\Delta_{uvw}$, $\Delta_{vzw}$, and $\Delta_{zuw}$ contain in their interior subsets of $P\setminus \{p\}$ with $n_{uv}-1$ points, with $n_{vz}$ points, and with $n_{zu}$ points, respectively, and that by construction $p$ lies in $\Delta_{uvw}$, which hence contains $n_{uv}$ points of $P$.
\item We prove Property P5. Consider any curve $s\in S$. If $s$ is not incident to $p$, then it satisfies the property by induction. Otherwise, $s$ is incident to $p$, hence it has at least one of its end-points inside $\cl(\Delta_{uvw})$. Thus, we only need to show that, if its second end-point is inside $\cl(\Delta_{uvw})$, then $s$ is entirely contained in $\cl(\Delta_{uvw})$. Since $s$ crosses each of $s_{uw}$ and $s_{vw}$ at most once, it follows that $s$ is not entirely contained in $\cl(\Delta_{uvw})$ if and only if it crosses each of $s_{uw}$ and $s_{vw}$ {\em exactly} once. Suppose that $s$ connects $p$ with a point $u_i$. If $1\leq i\leq x$, then $s$ is arbitrarily close to $s_{uv}$, hence it does not cross $s_{uw}$ nor $s_{vw}$ at all. If $x+1\leq i\leq y$, then $s$ is composed of a curve arbitrarily close to $(u_1,u_2)$, which does not cross $s_{uw}$ nor $s_{vw}$ at all, and of a curve coinciding with part of edge $e_i$. If the latter curve crosses both $s_{uw}$ and $s_{vw}$, then $e_i$ would have both its end-points inside $\cl(\Delta_{uvw})$ and still would not entirely lie inside $\cl(\Delta_{uvw})$, which is not possible by induction.
\end{itemize}

This concludes the proof of the claim and of the lemma.
\end{proof}

Fig.~\ref{fig:plane3treesdrawing}(b) shows a planar drawing of $G_2$ satisfying
the properties of Lemma~\ref{le:splittingpoints}. Lemma~\ref{le:splittingpoints}
implies a proof of Theorem~\ref{th:plane3trees-main}. Namely, construct any
planar drawing $\Gamma_1$ of $G_1$. Denote by $P$ the point set to which the
$n-3$ internal vertices of $G_1$ are mapped in $\Gamma_1$. Let $s_{uv}$,
$s_{vz}$, and $s_{zu}$ be the curves representing edges $(u_1,v_1)$,
$(v_1,z_1)$, and $(z_1,u_1)$ in $\Gamma_1$, respectively. Let $S$ be the set of
curves representing the internal edges of $G_1$ in $\Gamma_1$. Let $p_u$,
$p_v$, and $p_z$ be the points on which $u_1$, $v_1$, and $z_1$ are drawn,
respectively. Let $R=\{p_u, p_v, p_z\}$. Construct a planar drawing $\Gamma_2$
of $G_2$ satisfying the properties of Lemma~\ref{le:splittingpoints}. Then,
$\Gamma_1$ and $\Gamma_2$ are planar drawings of $G_1$ and $G_2$, respectively.
By Properties (a) and (c) of Lemma~\ref{le:splittingpoints}, the $n$ vertices of
$G_2$ are mapped to the same $n$ points to which the vertices of $G_1$ are
mapped. Finally, by Properties (b) and (d) of Lemma~\ref{le:splittingpoints}, if
edges $e_1$ of $G_1$ and $e_2$ of $G_2$ have their end-vertices mapped to the
same two points $p_a,p_b\in P \cup\{p_u, p_v, p_z\}$, then $e_1$ and $e_2$ are
represented by the same Jordan curve in $\Gamma_1$ and in $\Gamma_2$; hence,
$\Gamma_1$ and $\Gamma_2$ are a \sefenomap of $G_1$ and $G_2$.

\section{Conclusions} \label{se:conclusions}

In this paper we studied the problem of determining the largest $k_1\leq n$ such that every $n$-vertex planar graph and every $k_1$-vertex planar graph admit a \sefenomap. We proved that $k_1\geq n/2$. No upper bound smaller than $n$ is known. Hence, tightening this bound (and in particular proving whether $k_1=n$ or not) is a natural research direction.

To achieve the above result, we proved that every $n$-vertex plane graph has an $(n/2)$-vertex induced outerplane graph, a result related to a famous conjecture stating that every planar graph contains an induced forest with half of its vertices~\cite{ab-cpg-79}. A suitable triangulation of a set of nested $4$-cycles shows that $n/2$ is a tight bound for our algorithm, up to a constant. However, we have no example of an $n$-vertex plane graph whose largest induced outerplane graph has less than $2n/3$ vertices (a triangulation of a set of nested $3$-cycles shows that the $2n/3$ bound cannot be improved). The following question arises: What are the largest $k_2$ and $k_3$ such that every $n$-vertex plane graph has an induced outerplane graph with $k_2$ vertices and an induced outerplanar graph with $k_3$ vertices? Any bound $k_2>n/2$ would improve our bound for the \sefenomap problem, while any bound $k_3>3n/5$ would improve the best known bound for Conjecture~\ref{conj:induced-forests}, via the results in~\cite{h-iftog-90}.

A different technique to prove that every $n$-vertex planar graph and every $k_4$-vertex planar graph have a \sefenomap is to ensure that a mapping between their vertex sets exists that generates no shared edge. Thus, we ask: What is the largest $k_4\leq n$ such that an injective mapping exists from the vertex set of any $k_4$-vertex planar graph and the vertex set of any $n$-vertex planar graph generating no shared edge? It is easy to see that $k_4\geq n/4$ (a consequence of the four color theorem~\cite{ah-epfci-77,ahk-epfcii-77}) and that $k_4\leq n-5$ (an $n$-vertex planar graph with minimum degree $5$ does not admit such a mapping with an $(n-4)$-vertex planar graph having a vertex of degree $n-5$).

Finally, it would be interesting to study the geometric version of our problem. That is: What is the largest $k_5\leq n$ such that every $n$-vertex planar graph and every $k_5$-vertex planar graph admit a geometric simultaneous embedding with no mapping? Surprisingly, we are not aware of any super-constant lower bound for the value of $k_5$.

\bibliography{bibliography}
\bibliographystyle{plain}

\end{document}